\documentclass[aps,prx,superscriptaddress,nofootinbib,twocolumn]{revtex4-2}

\usepackage[utf8]{inputenc}
\usepackage{amsmath,amssymb,amsthm}
\usepackage{braket}
\usepackage{bm}
\usepackage{graphicx,color}
\usepackage[colorlinks=true,citecolor=blue,linkcolor=magenta]{hyperref}
\usepackage{url}
\usepackage{tabularx}
\usepackage{algorithm}    % For the algorithm environment
\usepackage{algpseudocode} % For pseudocode formatting
\usepackage{float}
\usepackage{physics}
\usepackage[capitalize]{cleveref}
\usepackage[normalem]{ulem}

\newtheorem{lemma}{Lemma}
\newtheorem{theorem}{Theorem}

\makeatletter
\newsavebox{\@brx}
\newcommand{\llangle}[1][]{\savebox{\@brx}{\(\m@th{#1\langle}\)}%
  \mathopen{\copy\@brx\kern-0.5\wd\@brx\usebox{\@brx}}}
\newcommand{\rrangle}[1][]{\savebox{\@brx}{\(\m@th{#1\rangle}\)}%
  \mathclose{\copy\@brx\kern-0.5\wd\@brx\usebox{\@brx}}}
\makeatother

\newcommand{\kket}[1]{
    \ensuremath{|{#1}\rrangle}
}

\newcommand{\bbraket}[1]{
    \ensuremath{\llangle{#1}\rrangle}
}

\DeclareMathOperator{\Var}{Var}

\crefname{section}{Section}{Sections}
\crefname{appendix}{Appendix}{Appendices}

\begin{document}

\title{Low-Resource Quantum Energy Gap Estimation via Randomization}

\author{Hugo Pages}
\email{hugo.pages@etu.unistra.fr}
\affiliation{%
    Graduate School of Engineering Science, The University of Osaka, 1-3 Machikaneyama, Toyonaka, Osaka 560-8531, Japan.
}
\affiliation{%
    Télécom Physique Strasbourg, Strasbourg University,
    300 Bd Sébastien Brant, Illkirch Graffenstaden, 67412, France.
}

\author{Chusei Kiumi}
\affiliation{%
    Center for Quantum Information and Quantum Biology, The University of Osaka, 1-2 Machikaneyama, Toyonaka, Osaka,  560-0043, Japan.
}%
\author{Yuto Morohoshi}
\affiliation{%
    Graduate School of Engineering Science, The University of Osaka, 1-3 Machikaneyama, Toyonaka, Osaka 560-8531, Japan.
}
\author{Bálint Koczor}
\affiliation{%
    Mathematical Institute, University of Oxford, Woodstock Road, Oxford OX2 6GG, United Kingdom.
}
\author{Kosuke Mitarai}
\affiliation{%
    Graduate School of Engineering Science, The University of Osaka, 1-3 Machikaneyama, Toyonaka, Osaka 560-8531, Japan.
}
\affiliation{%
    Center for Quantum Information and Quantum Biology, The University of Osaka, 1-2 Machikaneyama, Toyonaka, Osaka,  560-0043, Japan.
}
\email{mitarai.kosuke.es@osaka-u.ac.jp}

\begin{abstract}
    Estimating the energy spectra of quantum many-body systems is a fundamental task in quantum physics, with applications ranging from chemistry to condensed matter. Algorithmic shadow spectroscopy is a recent method that leverages randomized measurements on time-evolved quantum states to extract spectral information. However, implementing accurate time evolution with low-depth circuits remains a key challenge for near-term quantum hardware. In this work, we propose a hybrid quantum-classical protocol that integrates Time Evolution via Probabilistic Angle Interpolation (TE-PAI) into the shadow spectroscopy framework. TE-PAI enables the simulation of time evolution using shallow stochastic circuits while preserving unbiased estimates through quasiprobability sampling. We construct the combined estimator and derive its theoretical properties. Through numerical simulations, we demonstrate that our method accurately resolves energy gaps and exhibits enhanced robustness to gate noise compared to first-order Trotter-based shadow spectroscopy. We further validate the protocol experimentally on up to 20 qubits using IBM quantum hardware. This makes TE-PAI shadow spectroscopy a promising tool for spectral analysis on noisy intermediate-scale quantum (NISQ) devices.
\end{abstract}

\maketitle

\section{Introduction}

Quantum many-body systems are governed by the principles of quantum mechanics. A fundamental aspect of understanding these systems is the characterization of their energy spectra. In particular, the differences between eigenvalues of the Hamiltonian, often referred to as energy gaps, play an important role in determining their physical properties \cite{sachdev2011quantum}. For instance, in quantum chemistry, excitation energies are directly related to molecular optical responses and chemical reactivity \cite{levine2009quantum}. However, accurately computing energy gaps remains a challenging task for classical computers, primarily because the dimension of the Hilbert space that must be taken into account grows exponentially with the system size \cite{Dirac1929}.
Quantum computers provide a promising alternative, as they are capable of simulating quantum systems efficiently with polynomially scaling resources with respect to the system size~\cite{Lloyd1996}---and indeed the earliest examples of quantum advantage are expected to be in quantum simulation~\cite{zimboras2025myths,babbush2025grand}.

One of the earliest algorithms developed for estimating energy spectra on a quantum computer is quantum phase estimation (QPE) \cite{Kitaev1995, NielsenChuang2010}. It is often considered the main method within the framework of fault-tolerant quantum computing, as it can extract eigenvalues of a Hamiltonian with high precision. However, QPE requires deep circuits with ancillary qubits and error correction, which are far beyond the capabilities of current noisy intermediate-scale quantum (NISQ) devices. The limited qubit counts, short coherence times, and high gate error rates in NISQ hardware make direct implementation of QPE infeasible in the near term, motivating the exploration of alternative approaches that are better suited to present-day hardware.

As such, researchers have proposed several quantum algorithms for estimating spectral properties, aiming for approaches that reduce circuit depth and resource requirements while remaining compatible with the constraints of current hardware.
One promising method among them is algorithmic shadow spectroscopy \cite{Chan2025AlgorithmicShadow}. In shadow spectroscopy, we apply the time evolution operator of a target Hamiltonian to qubits for various durations and perform randomized measurements at each time point to generate classical shadows. By collecting and analyzing these time-dependent signals, one can reconstruct spectral information through classical post-processing. A notable advantage of this method is that it requires only straightforward time evolution of the Hamiltonian without the need for ancilla qubits or controlled time evolution.
Because of its simplicity, shadow spectroscopy offers better compatibility with near-term quantum devices compared to more demanding approaches such as quantum phase estimation \cite{Kitaev1995, NielsenChuang2010}. Also, by leveraging the classical shadow framework \cite{huang2020predicting}, it can estimate spectral properties with relatively few measurement shots, efficiently extracting relevant information from limited data. This feature provides a potential advantage over variational approaches such as the variational quantum eigensolver (VQE), which often require extensive optimization and large numbers of measurements to achieve comparable accuracy. However, this method still relies on precise time evolution, which can make its implementation challenging depending on the capabilities of the quantum hardware. One approach to address this limitation is to apply quantum error mitigation techniques \cite{Quantum_Error_Mitigated_Classical_Shadows}, such as probabilistic error cancellation, zero-noise extrapolation, or symmetry verification \cite{Endo2018QEM}. These approaches, however, require a detailed characterization of the noise and often lead to an exponential increase in the number of measurements.

To overcome this issue and further reduce the quantum resource requirements for spectral estimation, we propose an algorithm that integrates the recently developed technique of time-evolution via probabilistic angle interpolation (TE-PAI) into the framework of algorithmic shadow spectroscopy \cite{Kiumi2025}.
TE-PAI is a stochastic method for implementing time evolution under a given Hamiltonian based on Trotter decomposition.
Instead of applying all gates in the Trotterized sequence, TE-PAI probabilistically samples a subset of rotation gates according to carefully chosen quasiprobability distributions \cite{Kiumi2025, Koczor2024ProbabilisticInterpolation,PRXQuantum.5.040352}.
This approach allows us to reproduce the effect of a long and accurate Trotter sequence using only a small number of quantum gates, significantly reducing circuit depth while preserving fidelity. Furthermore, TE-PAI requires only Pauli rotation gates with fixed angles of $\Delta$ and $\pi$, followed by measurement. This approach, similar to shadow spectroscopy, does not rely on ancillary qubits or controlled time evolution. The simplicity and fundamental similarity of these randomized techniques make them naturally compatible and well aligned for integration. Although there is a trade-off between the measurement overhead and the number of gates, it is well understood and can be controlled explicitly \cite{Kiumi2025}: Tuning the single parameter $\Delta$ allows a user to increase or decrease the expected circuit depth at the cost of increasing the measurement overhead. Thus, the significant advantage of shadow spectroscopy---efficiently extracting spectral information with fewer shots---alleviates the increased shot requirement of TE-PAI.

By combining TE-PAI with algorithmic shadow spectroscopy, we construct an effective protocol for estimating energy spectra. Both TE-PAI and classical shadows involve sampling from randomized procedures, that is, one over gate sequences and the other over measurement bases. A key theoretical contribution of this work is the formulation of the combined estimator, which involves a nested quasiprobability sampling structure. This analysis ensures that the full protocol yields unbiased estimators of spectral quantities and clarifies the statistical properties of the resulting estimates. 

Since our algorithm is theoretically shown to offer a significant reduction in quantum resource requirements, we further highlight its practicality through extensive numerical simulations, which form the main contribution of this work, and demonstrative experiments on existing quantum hardware. Thanks to the simplicity of its circuit construction, our algorithm has the potential to serve as a key enabler of practical quantum advantage for NISQ and early fault-tolerant quantum computers (early-FTQC)~\cite{Kim2023Utility}.

While ref.~\cite{Kiumi2025} focused on the unique advantages of TE-PAI for early fault tolerant machines, we demonstrate a direct implementation of our algorithm on real quantum hardware through the simulation of the transverse Ising model. We benchmark its performance against conventional algorithmic shadow spectroscopy on a system of 20 qubits. Our results confirm the practical viability of our approach on contemporary noisy quantum devices, highlighting its potential for simulating many-body quantum systems within the current NISQ era. \\

This manuscript is organised as follows. \Cref{sec:preliminaries} introduces the theoretical background. In \cref{sec:te-pai-shadows}, we present the TE-PAI shadow spectroscopy protocol and provide a formal analysis of its unbiasedness and variance properties. \Cref{sec:experiments} reports numerical simulations, and the following subsection reports experimental results on quantum hardware, benchmarking the proposed method against standard Trotter-based shadow spectroscopy. Finally, in \cref{sec:conclusion}, we conclude with a discussion of the results, practical implications, and future research directions. Additional technical details and proofs are provided in the appendices.

\section{Preliminaries}\label{sec:preliminaries}
In this section, we introduce the necessary background for understanding our approach. Throughout this paper, we write boldface for quantum channels acting on density matrices.
For a unitary $U$, the associated channel $\bm{U}$ has the action $\bm{U}(\rho)=U\rho U^{\dagger}$ for any density matrix $\rho$.
For a $d\times d$ matrix $\rho = \sum_{i,j}\rho_{ij}\ket{i}\bra{j}$, we also introduce a vectorized notation as $\kket{\rho}:=\sum_{i,j=1}\rho_{ij}\ket{i}\otimes\ket{j}.$
With this convention, a channel corresponding to a unitary $U$ can be written as $\kket{\bm{U}(\rho)}=(U\otimes U^{*})\kket{\rho}$.
We also introduce the Hilbert-Schmidt inner product between two operators $A$ and $B$ as $\bbraket{A|B}:=\mathrm{Tr}(A^{\dagger} B)$ in accordance with the Liouville space representation \cite{Gilchrist2005Vectorization}.

Our central goal is to enable the estimation of transition energy gaps from time-evolved classical shadow data using resource-efficient circuits optimized for the constraints of NISQ and early-FTQC quantum hardware.

\subsection{Energy Gap Estimation}
Consider an $N$-qubit Hamiltonian $H=\sum_{j=1}^{J} h_j P_j$, where each $P_j \in \{I,X,Y,Z\}^{\otimes N}$ is a Pauli string. The time evolution operator generated by $H$ for a duration $t$ is denoted by $U_{H,t}:=e^{-iHt}$. Letting $\{E_\alpha\}$ and $\{|E_\alpha\rangle\}$ denote the eigenvalues and eigenstates of $H$, respectively, we define the transition energy gap between levels $\alpha$ and $\beta$ as:
\[
    \Delta E_{\alpha\beta} := |E_\alpha - E_\beta|.
\]
This definition includes the spectral gap as a special case. When the eigenvalues are ordered as $E_0 \leq E_1 \leq \cdots$, the spectral gap usually refers to $E_1-E_0$, which plays an important role in physical phenomena such as phase transitions, thermalization timescales, and adiabatic evolution rates. In this work, we consider the more general problem of estimating transition energy gaps $\Delta E_{\alpha\beta}$, which appear as frequencies in dynamical signals. Which transition gaps are visible depends on the initial state and on the measured observable.

In practice, quantum algorithms for gap estimation rely on extracting dynamical information from a time-evolved quantum state. Suppose a system with an initial state $\ket{\psi_{\text{init}}} = \sum_{\alpha} c_\alpha \ket{E_\alpha}$ evolves under $H$ for time $t$ as $\ket{\psi(t)} = e^{-i H t} \ket{\psi_{\text{init}}}= \sum_{\alpha} c_\alpha e^{-i E_\alpha t} \ket{E_\alpha}.$
The expectation value of an observable $O$ at time $t$ becomes
\begin{align}\label{eq:time_signal}
    S(t)
    &:= \bra{\psi(t)} O \ket{\psi(t)}\nonumber
    \\
    &= \sum_{\alpha,\beta} c_\alpha^* c_\beta
    e^{i (E_\alpha - E_\beta) t} \bra{E_\alpha} O \ket{E_\beta}.
\end{align}
The oscillatory terms in $S(t)$ have frequencies equal to the energy differences $\Delta E_{\alpha\beta}$, and thus a Fourier transform of the time signal reveals peaks at these gaps.
The intensity of each peak is determined by the spectral weight $c_\alpha^*c_\beta\langle E_\alpha|O|E_\beta\rangle$. Therefore, a prominent peak need not correspond to an adjacent-level spacing, but rather to a transition with large weight for the chosen initial state and observable.

Exact determination of $\Delta E_{\alpha\beta}$ requires sufficiently precise and long time evolution to achieve the desired frequency resolution. On NISQ devices, however, implementing long and accurate time evolutions is challenging due to limited coherence times and gate imperfections.

\subsection{Time Evolution}
From \cref{eq:time_signal}, efficiently simulating the time evolved state $\kket{\rho(t)}:=\kket{\bm{U}_{H,t}(\rho_{\text{init}})}$ and estimating expectation values of many observables $O$ denoted as
\[S(t) =\bbraket{O|\rho(t)}\] 
are key components of energy gap estimation algorithms in quantum computing. Since our target is NISQ and early-FTQC devices, we focus on methods that can implement time evolution with low quantum resource requirements. 

One of the most straightforward approaches is Trotterization \cite{Lloyd1996,Suzuki1,Suzuki2}, which approximates $U_{H,t}$ as a product of exponentials of individual terms in the Hamiltonian. Let $K$ denote the number of Trotter steps and define the step size $\delta t := t/K$. The $K$-step first-order Trotter circuit is denoted by $U_{H,t}^{(K)}$ and is given by \begin{equation}\label{eq:trotter_circuit}
    U_{H,t}^{(K)} := \left[\prod_{j=1}^J e^{-i h_jP_j \delta t}\right]^K
    = \left[\prod_{j=1}^J R_{P_j,\theta_j}\right]^K,
\end{equation}
where $\theta_j= 2h_j\delta t$ and $R_{P,\theta}:=\exp(-i\theta P/2)$ is a Pauli rotation gate. The expectation value obtained by evolving $\rho_{\rm init}$ with the $K$-step Trotter circuit $\langle O \rangle^{(K)}_{t} := \bbraket{O|\bm{U}_{H,t}^{(K)}(\rho_{\rm init})}$ approximates the true expectation value as $\epsilon_T:=\left| S(t) - \langle O \rangle^{(K)}_t \right|\le \mathcal{O}(t^2/K).$
Also, the sample complexity scaling as $\mathcal{O}\!\left( \epsilon_s^{-2} \right)$ is required to estimate $\langle O \rangle^{(K)}_{t}$ within an additive statistical error $\epsilon_s$.

We can further reduce the circuit depth required for time evolution by employing the recently proposed method of time evolution via probabilistic angle interpolation (TE-PAI) \cite{Kiumi2025,Koczor2024ProbabilisticInterpolation}. This method is based on quasiprobability sampling \cite{Temme2017PEC,Endo2018PracticalQEM} which we summarize in \cref{app:quasiprobability}.

 \subsection{Time-Evolution by Probabilistic Angle Interpolation (TE-PAI)} \label{app:TE-PAI}
 This section briefly reviews the result on TE-PAI method proposed in Ref.~\cite{Kiumi2025}, including the asymptotic analysis on expected gate count and variance. TE-PAI method applies probabilistic angle interpolation (PAI) \cite{Koczor2024ProbabilisticInterpolation} to Pauli rotation gates.

For any $\theta\in(-\pi,\pi)$ and a fixed angle $\Delta\in[|\theta|,\pi)$, set $\phi=\mathrm{sign}(\theta)\Delta$. A superoperator representation of Pauli rotation $R_{P,\theta}=\exp(-i\theta P/2)$ admits the decomposition
\begin{align}
\label{eq:te-pai-decomp-prelim-app}
\bm{R}_{P,\theta}=a_1(\theta)\bm{R}_{P,1}+a_2(\theta)\bm{R}_{P,2}+a_3(\theta)\bm{R}_{P,3},
\end{align}
where $\bm{R}_{P,1},\bm{R}_{P,2},\bm{R}_{P,3}$ are superoperator representations of Pauli rotation gates $I,  {R}_{P,\phi},{R}_{P,\pi}$ and the analytic expression of the coefficients is explicitly given as follows:
\begin{align*}
    \begin{split}
        a_1(\theta) &=
        \csc \left(\frac{\Delta}{2}\right)
        \cos \left(\frac{|\theta|}{2}\right)
        \sin \left(\frac{\Delta}{2} - \frac{|\theta|}{2}\right),  \\
        a_2(\theta) &=
        \csc(\Delta) \sin(|\theta|),  \\
        a_3(\theta) &=
        -\sec \left(\frac{\Delta}{2}\right)
        \sin \left(\frac{|\theta|}{2}\right)
        \sin \left(\frac{\Delta}{2} - \frac{|\theta|}{2}\right).
    \end{split}
\end{align*}

Let $\gamma(\theta):=|a_1(\theta)|+|a_2(\theta)|+|a_3(\theta)|$ be the normalization factor. We define the random rotation channel outputting one of these three channels as $\bm{R}_{P,l}$, where $l\in\{1,2,3\}$ is sampled with probabilities $p_l(\theta):=|a_l(\theta)|/\gamma(\theta)$. Fixing $\Delta$ and applying the random-rotation construction independently to each Pauli-rotation gate $R_{P_j,\theta_j}$ in the circuit $U_{H,t}^{(K)}$ in Eq.~\eqref{eq:trotter_circuit}, yields a random TE-PAI circuit
\begin{equation}
\label{eq:def_random_te_pai_circuit}
\hat{\bm{U}}_{\bm{l},t}^{(K)}
:=\prod_{k=1}^{K}\prod_{j=1}^{J}\bm{R}_{P_j,l(j,k)},
\end{equation}
along with an accumulated classical weight
$
\Gamma_{\bm{l}}:=\Gamma \prod_{k=1}^{K}\prod_{j=1}^{J}\mathrm{sign}(a_{l(j,k)}),
$ where $\Gamma:=\left(\prod_{j=1}^{J}\gamma(\theta_j)\right)^K$.
Here, $l(j,k)\in\{1,2,3\}$ is sampled independently according to probability $p_{l}(\theta_j)$ and by definition, $\Gamma_{\bm{l}}^2=\Gamma^2$. 

Since this is quasiprobability decomposition by construction, the following random map is an unbiased estimator of the ideal rotation channel $\bm{R}_{P,\theta}$:
\begin{equation*}
\mathbb{E}\!\left[\gamma(\theta)\,\mathrm{sign}(a_l(\theta))\,\bm{R}_{P,l}\right]=\bm{R}_{P,\theta},
\end{equation*}

Subsequently, the TE-PAI channel \cref{eq:def_random_te_pai_circuit} along with classical weight form an unbiased estimator of the Trotterized time evolution channel as:
\begin{equation*}
\mathbb{E}\!\left[\Gamma_{\bm{l}} \hat{\bm{U}}_{\bm{l},t}^{(K)}\right]=\bm{U}_{H,t}^{(K)}.
\end{equation*}

  Let $\hat{o}_{\bm{l}}$ be a single-shot measurement estimator of the observable $O$ obtained by executing $\hat{\bm{U}}_{\bm{l},t}^{(K)}$ on $\rho_{\rm init}$. The multiplication of $\Gamma_{\bm{l}}$ in post-processing yields an unbiased estimator of the original Trotter expectation value:
\begin{equation}
\mathbb{E}\left[\Gamma_{\bm{l}}\hat{o}_{\bm{l}}\right]=\langle O\rangle^{(K)}_t\approx S(t) ,
\end{equation}
which approximates the true expectation value for sufficiently large $K$.
 
In implementation, each gate $R_{P_j,\theta_j}$ is replaced by one of $\{I,R_{P_j,\phi},R_{P_j,\pi}\}$ and sampling $I$ corresponds to skipping the gate. 

As $K$ increases, the rotation angle $\theta_j$ in each Pauli rotation becomes smaller, so the probability of sampling the identity gate approaches one. This mechanism probabilistically skips most operations, yielding a sampled circuit with depth $\nu_{\bm{l}} \ll JK$, much shallower than the original Trotter sequence. Consequently, the expected number of non-identity gates $\mathbb{E}[\nu_{\bm{l}}]$ scales linearly with $t$ and converges to a constant in the large-$K$ limit:
\begin{equation*}
    \mathbb{E}[\nu_{\bm{l}}]
    =
    \csc(\Delta)\,(3-\cos\Delta)\|H\|_1 t
    + O(1/K),
\end{equation*}
where $\|H\|_1 := \sum_j |h_j|$ is the $L^1$ norm of the Hamiltonian coefficients. In contrast, the number of gates in the corresponding Trotter circuit scales as $O(t^2)$ and diverges as $K \to \infty$.

This depth reduction comes at the cost of a statistical sampling overhead, captured by the accumulated weight $\Gamma$, which leads to a sample complexity scaling of $\mathcal{O}(\Gamma^2\epsilon_s^{-2})$. Although the overhead factor $\Gamma^2$ scales exponentially with the evolution time $t$, it converges to a constant in the large-$K$ limit:
\begin{equation*}
    \Gamma^2
    =
    \exp\!\left[
        4t\|H\|_1
        \tan\!\left(\frac{\Delta}{2}\right)
    \right]
    + O(1/K).
\end{equation*}
The parameter $\Delta$ therefore provides a configurable trade-off between circuit depth and statistical sampling cost, allowing the overhead to remain manageable.

The advantage of TE-PAI lies in its tunable parameter $\Delta$, which explicitly governs the trade-off between circuit depth and sampling overhead. In the large-$K$ limit, increasing $\Delta$ reduces the expected circuit depth but increases the sampling overhead; decreasing $\Delta$ has the opposite effect, lowering the sampling cost at the expense of deeper circuits. As remarked in Ref.~\cite{Kiumi2025}, by setting $\Delta$ to a sufficiently small value defined by $\Delta=2 \arctan \left(\frac{Q}{4\|H\|_1 t}\right)$, TE-PAI achieves a constant measurement overhead $\exp(Q)$ while the gate count scales as $\mathcal{O}\left(\frac{\|H\|^2_1 t^2}{Q}\right)$, which is comparable to the standard Trotter approach.

Furthermore, Ref.~\cite{Kiumi2025} demonstrated that TE-PAI is highly advantageous for the early-FTQC regime. Since the only non-Clifford operation required is $R_Z(\Delta)$, the protocol offers the potential to significantly reduce the number of T-gates needed to execute a single circuit instance.

 Higher-order Trotter--Suzuki formulas improve the product-formula error scaling. In particular, symmetric second-order Trotterization achieves this improvement without increasing the gate count relative to two first-order steps~\cite{Suzuki1,Suzuki2,Childs_2021}. As discussed in Appendix~\ref{app:higher-order-te-pai}, TE-PAI can also be constructed from the second-order formula and inherits its improved Trotter error. Since our focus is the trade-off between gate count and sampling cost for a given product formula, we use the first-order formula as a simple common baseline for our comparisons.

\subsection{Shadow Spectroscopy}
While the sampling overhead of TE-PAI is manageable, it becomes a critical bottleneck when scaling to a large number of observables. To circumvent this, we integrate the classical shadow framework \cite{huang2020predicting}; by employing randomized measurements, this approach allows us to estimate exponentially many observables simultaneously with only logarithmic sample scaling. In particular, we utilize the Pauli-basis variant, which is well-suited for near-term hardware and enables the efficient reconstruction of local Pauli observables.

Given a quantum state $\rho$, we apply single-qubit Clifford unitaries $U_1,\ldots,U_N$ independently and uniformly at random, then measure in the computational basis to obtain a bitstring $b=b_1\cdots b_N$.
From this outcome we construct a snapshot
\begin{align}
\label{eq:shadow-snapshot-prelim}
\kket{\hat{\rho}}=\bigotimes_{n=1}^{N}\left(3\,\bm{U}_n^{-1}\kket{b_n}-\kket{I}\right).
\end{align}
It follows that $\mathbb{E}[\kket{\hat{\rho}}]=\kket{\rho}$, so for any observable $O$ the quantity $\bbraket{O|\hat{\rho}}$ is an unbiased estimator of $\mathrm{Tr}(O\rho)$.
By repeating the procedure many times, one obtains snapshots that allow simultaneous estimation of a large set of observables with rigorous sample complexity guarantees \cite{huang2020predicting}. The key notion for analyzing the sample complexity is the shadow norm, which quantifies the variance of the estimator for a given observable.
 We define the squared shadow norm $\|O\|_{\text{sh}}^2$ as the maximum second moment of the estimator over all possible states:
\begin{equation}
    \|O\|_{\text{sh}}^2 := \max_{\rho} \mathbb{E} \left[\bbraket{O|\hat{\rho}}^2 \right],
\end{equation}
where the expectation $\mathbb{E}$ is taken over the random unitary transformations and measurement outcomes inherent to the shadow protocol.

For the specific case of Pauli-basis measurements and observables that are $q$-local Pauli strings, the squared shadow norm is given by $3^q$~\cite{huang2020predicting}. The same collection of classical shadow snapshots can be reused to estimate all observables in the target set $\{O_i\}_{i=1}^{N_o}$. Hence the sample complexity is controlled by the shadow norm for a single observable, while requiring all $N_o$ estimates to be accurate simultaneously introduces only a union-bound factor logarithmic in $N_o$~\cite{huang2020predicting, Chan2025AlgorithmicShadow}. The overall sample complexity is therefore
\[
    \mathcal{O}\!\left(3^q \epsilon^{-2}\log N_o\right).
\]

Finally, we employ the framework of algorithmic shadow spectroscopy \cite{Chan2025AlgorithmicShadow}. We consider a sequence of time points $\{t_n\}_{n=1}^{N_t}$. Using the classical shadow snapshots $\{\hat{\rho}(t_n)\}$ given by \cref{eq:shadow-snapshot-prelim} from the state $\rho(t_n)$ at each instance and a set of target observables $\{O_i\}_{i=1}^{N_o}$, we construct the estimated time series $\hat{S}_i(t_n) = \bbraket{O_i | \hat{\rho}(t_n)}$. This procedure yields a total of $N_oN_t$ data points. Since the above simultaneous-estimation guarantee applies at each time point, the total sample complexity scales as $\mathcal{O}(3^q N_t\epsilon^{-2}\log N_o)$ to estimate all $N_o$ time series. Following standardization and statistical filtering, we compute the cross-correlations of these signals and perform a Fourier transform; the resulting spectral peaks reveal the transition energy gaps of the Hamiltonian. A detailed description of the process is provided in Appendix \ref{sec:shadow-spectroscopy}.

\section{TE-PAI shadow spectroscopy}\label{sec:te-pai-shadows}
The central contribution of this work is a resource-efficient algorithm that integrates TE-PAI with algorithmic shadow spectroscopy. Explicitly designed for current and near-future quantum hardware, our protocol features a highly configurable trade-off between circuit depth and sampling overhead, minimizing quantum resource requirements while maintaining the flexibility to adapt to specific hardware constraints.

\subsection{Our protocol}
The following is the overview of the full procedure of TE-PAI shadow spectroscopy:
\begin{enumerate}
    
    \item \textbf{Step 1: Generate TE-PAI Circuits.}
    Set the total Trotter steps as $K_{\text{step}}=K \times N_t$ with positive integer $K$ and the number of time points $N_t$, and decompose the time evolution operator $e^{-iHt}$ into Trotter circuit $U_{H,t_s}^{(sK)}$ for each time point $t_s= s \frac{t}{N_t}$. For each gate in the sequence, probabilistically replace it with a fixed-angle rotation \{$0,\Delta,\pi$\} to generate $M_{\text{TE-PAI}}$ TE-PAI circuits along with classical weights $\left(\Gamma_{\bm{l}_m,t_s}, \hat{U}_{\bm{l}_m,t_s}^{(sK)}\right)$ according to \cref{eq:def_random_te_pai_circuit}.

    \item \textbf{Step 2: Obtaining Shadow Snapshots.}
   Apply random single-qubit Clifford unitaries $\{U_n\}_{n=1}^N$ to the state $\hat{U}_{\bm{l}_m,t_s}^{(sK)}(\rho_{\text{init}})$, then measure in the computational basis to obtain a bitstring $b=b_1\cdots b_N$. We then classically store the snapshot as $(\Gamma_{\bm{l},t_s}, \{U_n\}_{n=1}^N, \{b_n\}_{n=1}^N)$.
 
    \item \textbf{Step 3: Classical Signal Processing.}
    For $N_o$ observables $\{O_i\}_{i=1}^{N_o}$, compute estimates of time series signals $S_i(s)=\bbraket{O_i|\rho(t_s)}$ from estimator $\Gamma_{\bm{l_m},t_s} \bbraket{O_i|\hat{\rho}_{\bm{l}_{m}}(t_s)}$ which is constructed from the stored data according to \cref{eq:shadow-snapshot-prelim}. These time series are then processed, as described in Appendix~\ref{sec:shadow-spectroscopy} to extract the energy spectrum of the Hamiltonian.
\end{enumerate}

\begin{figure}[t]
    \begin{algorithm}[H]
        \caption{Obtaining a shadow snapshot of $e^{-iHt}\rho_{\text{init}} e^{iHt}$ using TE-PAI\label{alg:te-pai-shadow}}
        \begin{algorithmic}[1]
            \Statex \textbf{Input:} $N$-qubit Hamiltonian $H = \sum_{j=1}^J h_j P_j$, initial state $\rho_{\text{init}}$, total time $t$, number of time steps $K_{\rm steps}$, angle $\Delta$
            \Statex \textbf{Output:} Tuple $(\Gamma, \{U_n\}_{n=1}^N, \{b_n\}_{n=1}^N)$, where $\Gamma\in\mathbb{R}$ is the post-processing weight, $U_n$ is the local Clifford gate applied to qubit $n$, and $b_n\in\{0,1\}$. The shadow is constructed as $\Gamma \bigotimes_{n=1}^N \left[3 \bm{U}_n^{-1} \kket{b_n} - \kket{I}\right]$.
            \State Set $\delta = t / K_{\rm steps}$
            \State Initialize $\Gamma \leftarrow 1$
            \State Initialize an empty quantum circuit $\mathcal{C}$
            \For{$k = 1$ to $K_{\rm steps}$}
                \For{$j=1$ to $J$}
                    \State Compute $\theta=2h_j\delta$
                    \State Compute $a_1,a_2,a_3$ according to Eq.~\eqref{eq:te-pai-decomp-prelim-app}
                    \State Set $\gamma = |a_1|+|a_2|+|a_3|$, sample $d\in\{1,2,3\}$ with probability $|a_d|/\gamma$
                    \State Update $\Gamma \leftarrow \Gamma \times \gamma \times \mathrm{sign}(a_d)$
                    \State Append gate $R_{P_j,\phi}$ ($\phi=\mathrm{sign}(\theta)\Delta$ if $d=2$) or $R_{P_j,\pi}$ ($d=3$) to $\mathcal{C}$; if $d=1$, append nothing
                \EndFor
            \EndFor
            \For{$n = 1$ to $N$}
                \State Sample a Clifford gate $U_n$ uniformly at random and add to qubit $n$
            \EndFor
            \State Apply $\mathcal{C}$ to $\rho_{\rm init}$, measure in the computational basis to obtain $b=(b_1,\ldots,b_N)$
            \State \Return $\Gamma$, $\{U_n\}$, $\{b_n\}$
        \end{algorithmic}
    \end{algorithm}
\end{figure}

The detailed pseudocode for Step 1 and Step 2 is provided in Algorithm~\ref{alg:te-pai-shadow}.
In our actual implementation, the process for generating these snapshots slightly differs from a simple repetition of the fundamental algorithm described in Algorithm \ref{alg:te-pai-shadow}.
Specifically, our implementation generates one time-evolution circuit via the TE-PAI process and reuses it with $N_s$ different randomized measurement bases to obtain $N_s$ snapshots for each TE-PAI sampled circuits.
Consequently, if we generate $M_{\text{TE-PAI}}$ distinct time-evolution circuits via TE-PAI sampling, the total number of unique quantum circuits ultimately executed is $M_{\text{TE-PAI}} \times N_s$.
The detailed pseudocode for this implementation is provided in Appendix~\ref{app:experimental_algorithm} as Algorithm~\ref{alg:te-pai-shadow-experimental}.

This modification allows us to control the total number of circuit executions directly, while independently varying $M_{\text{TE-PAI}}$ and $N_s$. We show in the next subsection that this modified implementation preserves the unbiased nature of the estimator.

\subsection{Formal Analysis}\label{sec:formal_analysis}
In this subsection, we analyze the statistical properties of the estimator obtained from TE-PAI shadow spectroscopy. For generality, we state the result in a form that applies not only to TE-PAI but also to any quasiprobability sampling method. 

We use the notation from \cref{app:quasiprobability}, where we consider a sequence of channels $\bm{C}=\bm{C}_K\cdots \bm{C}_1$ with each $\bm{C}_k$ decomposed as $\bm{C}_k=\sum_{d=1}^D a_{k,d}\bm{C}_{k,d}$. Let $l_k\in \{1,\ldots,D\}$ be a random variable drawn according to $p_{d}=|a_{k,d}|/\gamma(\bm{C}_k)$, where $\gamma(\bm{C}_k)=\sum_d |a_{k,d}|$. 
Define the random channel $\hat{\bm{C}}_{\bm{l}} := \bm{C}_{K,l_K} \cdots \bm{C}_{2,l_2}\bm{C}_{1,l_1}$
along with an accumulated classical weight
$
\Gamma_{\bm{l}}:=\Gamma \prod_{k=1}^{K}\mathrm{sign}(a_{k,l_k}),
$ where $\Gamma:=\prod_{k=1}^{K}\gamma(\bm{C}_{k})$. We generate $M$ independent samples of classical weight and circuit denoted as $\Gamma_{\bm{l}_m},\ \hat{\bm{C}}_{\bm{l}_m}$ with $m\in\{1,\dots,M\}$. For each sampled channel applied on the initial state $\rho$, we obtain $N_s$ independent classical shadow snapshots $\hat{\rho}_{m,s}$ ($s\in\{1,\dots,N_s\}$) according to \cref{eq:shadow-snapshot-prelim}. We define the estimator $o_{m,s} := \Gamma_{\bm{l}_m} \bbraket{O | \hat{\rho}_{m,s}}$ and the overall estimator as
    \begin{align}\label{eq:overall-estimator}
        \langle \hat{O}\rangle & := \frac{1}{M N_s}\sum_{m=1}^{M}\sum_{s=1}^{N_s} o_{m,s}.
    \end{align}
Then, the following theorem shows that the estimator $\langle \hat{O}\rangle$ is an unbiased estimator of the true expectation value $\bbraket{O|\bm{C}(\rho)}$. For the specific case of TE-PAI, $\bm{C}$ corresponds to the Trotterized time evolution channel $\bm{U}_{H,t}^{(K)}$ and the true expectation value is $\langle O \rangle^{(K)}_{t} $.
\begin{theorem}[Unbiasedness]\label{thm:tepai_unbiased_shadow}
The estimator defined in \cref{eq:overall-estimator} is an unbiased estimator of the true expectation value:
    \begin{equation*}
        \mathbb{E}\!\left[\langle \hat{O}\rangle\right] =  \bbraket{O|\bm{C}(\rho)},
    \end{equation*}
where the expectation $\mathbb{E}[\cdot]$ is taken over both the quasiprobability sampling and the randomized measurements for classical shadow.
\end{theorem}
The detailed proof is provided in Appendix~\ref{app:proof_of_unbiasedness}. This demonstrates that the estimator remains unbiased regardless of $M$ and $N_s$.  The convergence to this true value is guaranteed by the law of large numbers.
    
The variance of the estimator is dominated by the statistical fluctuations inherent in the quasiprobability sampling.
    Increasing $N_s$ reduces the variance from the shadow measurements for a fixed evolution circuit, but it does not reduce the variance arising from the choice of different evolution circuits.
    The latter is suppressed only by increasing the number of quasiprobability samples, $M$.
    Therefore, by increasing $M$, the estimator is guaranteed to converge to the true expectation value.

\begin{theorem}[Variance of the overall estimator]\label{thm:tepai_shadow_variance}
The variance of this estimator is upper bounded as
\begin{align*}
    \operatorname{Var}[\langle \hat{O} \rangle] &\leq \Gamma ^{2}\left(\frac{\| O\|^2 _{\operatorname{sh}}-\| O\|^2}{M N_s} +\frac{\| O\|^2}{M}\right).
\end{align*}
For the specific case of Pauli-basis measurements and normalized $q$-local Pauli observables, this becomes
\begin{align*}
    \mathrm{Var}[\langle \hat{O} \rangle] \leq \Gamma ^{2}\left(\frac{3^q-1}{M N_s} +\frac{1}{M}\right).
\end{align*}
\end{theorem}
The detailed proof is provided in Appendix~\ref{app:proof_of_variance}. 

By fixing the total number of circuit executions $N_{\text{total}}:=M N_s$, we can further simplify the variance bound as
\[ \mathrm{Var}[\langle \hat{O} \rangle] \leq \frac{\Gamma ^{2}}{N_{\text{total}}}\left(3^q-1 +N_s\right).\]
This bound reveals a favorable trade-off for experimental implementation: as long as the number of shots per circuit satisfies $N_s \lesssim 3^q$, the statistical penalty for reusing the same circuit instance is small. Consequently, we can reuse a single compiled TE-PAI circuit instance up to $N_s \ll 3^q$ times with negligible impact on the overall estimation error. This allows for significant amortization of circuit compilation and loading costs without compromising statistical precision. We numerically verify this behavior in Section~\ref{sec:classical_simulation}, demonstrating that the resulting spectrum depends mainly on the total shot budget.

\subsection{Resource comparison with prior approaches}

Energy-gap estimation can be performed by several established quantum algorithms, but many of them require heavier quantum resources or stronger assumptions than the setting considered here. Quantum phase estimation can estimate eigenphases with accuracy $\delta$, but it requires coherent controlled time evolution, ancilla registers, and long circuits whose size scales at least as $O(1/\delta)$~\cite{Kitaev1995,NielsenChuang2010}. We therefore focus on ancilla-free and control-free protocols. Variational and subspace-based approaches estimate gaps by preparing approximate low-energy states and measuring energies or effective Hamiltonian matrix elements, but their accuracy depends on the ansatz, state-preparation quality, and classical optimization, which may suffer from barren plateaus~\cite{Peruzzo2014,Cerezo2021,McClean2018BarrenPlateaus}. Another ancilla-free route is real-time spectroscopy, where one extracts gaps from the oscillation frequencies of time-dependent observables or correlation functions~\cite{matsuzaki2021direct,Gnatenko_2022}. Related ancilla-free and control-free approaches to energy and energy-gap estimation have also been developed~\cite{russo2021evaluating,kuji2025robust}. Our work builds on shadow spectroscopy in this real-time framework, using randomized classical-shadow measurements to estimate many observables over many time points~\cite{Chan2025AlgorithmicShadow,huang2020predicting}, and focuses on reducing the quantum cost of the required time evolution by replacing deterministic first-order Trotter circuits with TE-PAI sampled circuits.

Next, we compare our protocol with standard shadow spectroscopy combined with first-order Trotterized time evolution. Let $N_t$ be the number of time points, $N_o$ the number of observables, $q$ the locality of each observable, and $\epsilon$ the target additive accuracy for estimating each time-dependent observable. Standard shadow spectroscopy requires
\[
    N_{\rm meas}^{\rm SS}
    =
    O\!\left(3^q N_t \epsilon^{-2}\log N_o\right)
\]
randomized measurements~\cite{Chan2025AlgorithmicShadow,huang2020predicting}. This scaling describes the measurement cost, but not the cost of implementing the time evolution.

For a first-order Trotter formula, reducing the Trotter error to $\epsilon_{\rm Trot}$ requires
\[
    r
    =
    O\!\left(t^2\epsilon_{\rm Trot}^{-1}\right)
\]
Trotter steps, up to Hamiltonian-dependent constants~\cite{Suzuki1,Childs_2021}. If $G_{\rm step}$ denotes the gate count of one Trotter step, the gate count of the deterministic Trotter circuit therefore scales as
\[
    G_{\rm Trot}(t,\epsilon_{\rm Trot})
    =
    O\!\left(G_{\rm step}t^2\epsilon_{\rm Trot}^{-1}\right).
\]
Thus, suppressing the systematic Trotter bias requires increasingly deep circuits.

We first consider TE-PAI at a fixed finite Trotter step size. In this setting, TE-PAI can be viewed as a randomized implementation of the same first-order Trotterized evolution. It therefore has the same Trotter bias as the corresponding deterministic product formula, but reduces the expected gate count of each sampled circuit by probabilistically omitting rotations. If $G_{\rm samp}(t)$ denotes the expected gate count of a sampled TE-PAI circuit, then
\[
    G_{\rm samp}(t)
    \leq
    G_{\rm Trot}(t,\epsilon_{\rm Trot}).
\]
This gate-count reduction is obtained at the cost of a sampling overhead. If $\Gamma$ denotes the TE-PAI normalization factor, the effective measurement cost becomes
\[
    N_{\rm meas}^{\rm TE\mbox{-}PAI+SS}
    =
    O\!\left(
    \Gamma^2 3^q N_t \epsilon^{-2}\log N_o
    \right).
\]
Hence, at finite Trotter step size, TE-PAI trades expected gate count per circuit for a multiplicative sampling overhead $\Gamma^2$, while retaining the same first-order Trotter bias.

The stronger advantage appears when the zero-Trotter-error limit is considered. In deterministic first-order Trotterization, taking $\epsilon_{\rm Trot}\to 0$ makes the circuit cost diverge as
\[
    G_{\rm Trot}(t,\epsilon_{\rm Trot})
    =
    O\!\left(G_{\rm step}t^2\epsilon_{\rm Trot}^{-1}\right).
\]
By contrast, the infinite-step TE-PAI limit gives an unbiased estimator of the exact time evolution. In this limit, the expected number of sampled elementary rotations scales only linearly with the simulated time,
\[
    G_{\rm TE\mbox{-}PAI}(t)
    =
    O(t),
\]
up to Hamiltonian-dependent constants and the cost per sampled term. Therefore,
\[
    G_{\rm TE\mbox{-}PAI}(t)
    \ll
    G_{\rm Trot}(t,\epsilon_{\rm Trot})
    \qquad
    \mbox{as } \epsilon_{\rm Trot}\to 0.
\]
Thus, TE-PAI has a two-level advantage: at finite Trotter step size it reduces the expected gate count while preserving the same discretization bias, and in the infinite-step limit it removes the Trotter bias entirely, replacing the accuracy-dependent $O(t^2\epsilon_{\rm Trot}^{-1})$ first-order Trotter gate scaling by an unbiased sampled evolution with linear-in-time expected gate scaling at fixed $\Delta$ and sampling overhead $\Gamma^2$.

TE-PAI is also particularly relevant for early fault-tolerant quantum computing. In this regime, a major cost comes from implementing non-Clifford rotations, usually through magic-state resources or $T$-gate synthesis. TE-PAI directly reduces this cost by decreasing the number of rotations that must be implemented in each sampled circuit~\cite{Kiumi2025}. This feature is especially compatible with recent proposals for efficient analog rotation gates, such as STAR-magic mutation~\cite{Toshio2026STARMagicMutation}. These approaches are most effective when many logical rotations share the same angle, because the cost of preparing and reusing the corresponding rotation resource can be amortized. TE-PAI naturally produces this favorable structure: instead of requiring a large set of unrelated rotation angles, it repeatedly samples rotations from a small set of fixed angles. Thus, in early-FTQC architectures, TE-PAI can reduce both the number of implemented rotations and the effective $T$-gate or magic-state compilation overhead.
\\
\section{Experiments}\label{sec:experiments}
In this section, we validate the performance of our proposed TE-PAI shadow spectroscopy protocol through a series of experiments conducted numerically and on quantum devices.
We benchmark our method against the standard algorithmic shadow spectroscopy~\cite{Chan2025AlgorithmicShadow}, using first-order Trotterized time evolution.
Hereafter, we refer to the spectra obtained by the standard and proposed methods as the ``Trotter-based spectrum'' and the ``TE-PAI-based spectrum,'' respectively.

\subsection{Classical simulations}
\label{sec:classical_simulation}
\begin{figure}
    \centering
\includegraphics[width=0.48\textwidth]{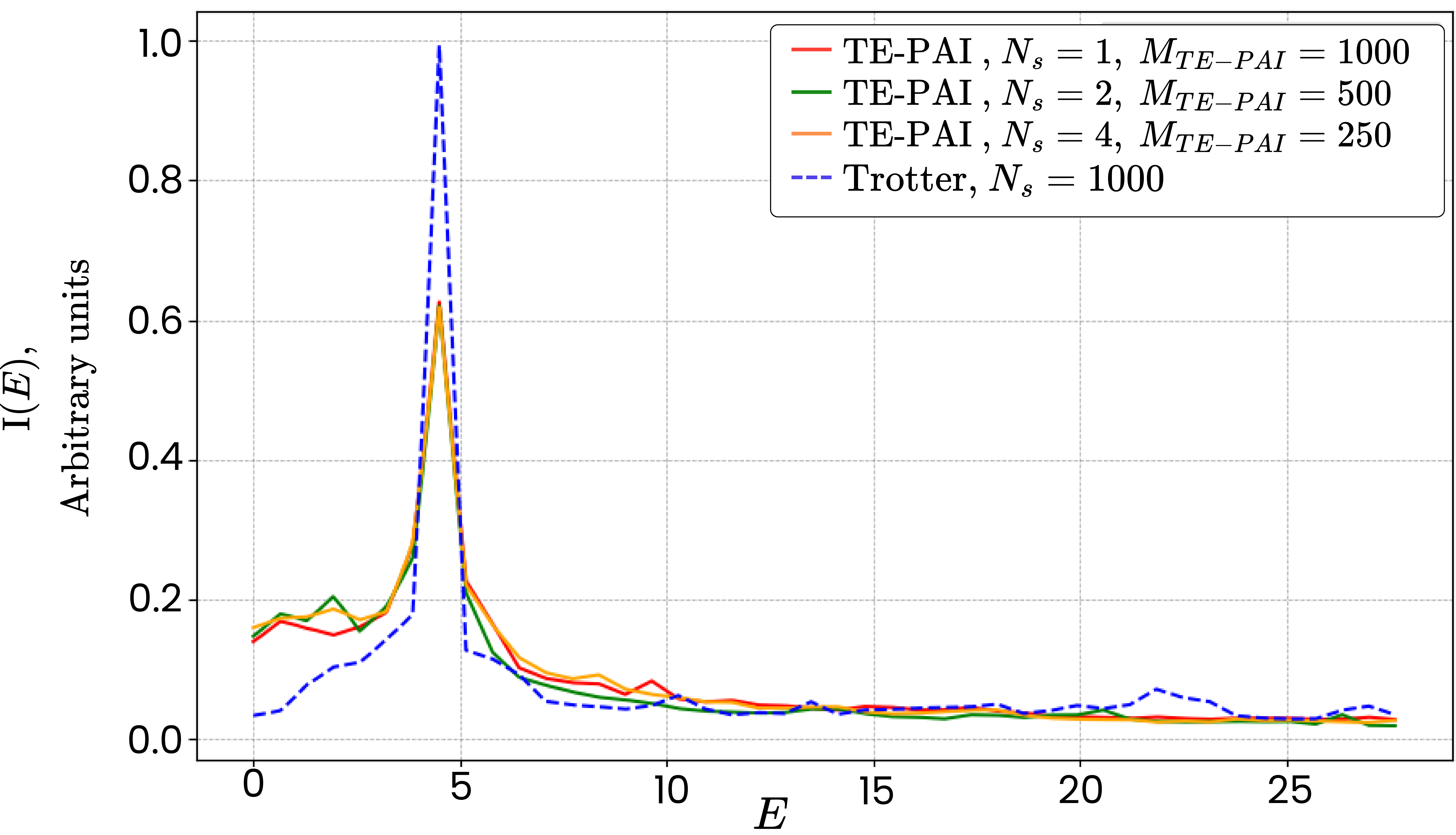}
    \caption{
        Energy spectra of the 10-qubit Heisenberg model from our TE-PAI-based method (solid lines) and the Trotter-based method (dotted line), with the total number of circuit executions fixed at 1000.
        The legend details the combinations of TE-PAI samples ($M_{\text{TE-PAI}}$) and shadow snapshots ($N_s$).
        For the same execution cost, the Trotter-based method shows a higher peak intensity, while the TE-PAI results are consistent across different configurations.
    }
    \label{fig_TE_PAI_shadow_spectro_same_circ_exec}
\end{figure}
In this subsection, we demonstrate the efficacy and robustness of TE-PAI shadow spectroscopy through classical simulations. 
We first conduct a noise-free simulation to validate that our protocol can accurately recover the energy gaps of a target Hamiltonian under ideal conditions.
Subsequently, we introduce a depolarizing noise model to assess the protocol's resilience to gate errors, a critical factor for practical applications on near-term devices. 
Through these simulations, performed on a Heisenberg model using the Qiskit Aer simulator, we compare the performance of our method against the conventional Trotter-based approach,  highlighting the advantages of TE-PAI in noisy environments.

For the noise-free simulations, we consider a 10-qubit one-dimensional Heisenberg model.
The Hamiltonian used in our experiments is
\begin{align}
    H = \sum_{i=1}^{N-1} \left( J_x X_i X_{i+1} +  J_y Y_i Y_{i+1} +  J_zZ_i Z_{i+1} \right),
\end{align}
where $X_i, Y_i, Z_i$ are the Pauli operators acting on the $i$-th qubit.  $J_x$, $J_y$ and $J_z$ are the coupling strength. They are set to 1 in our simulation. 

To isolate a specific transition energy gap, the initial state is prepared as a superposition of the ground state and the 10th excited state.
With this choice of initial state, the dominant oscillation frequency in the dynamical signal is expected to correspond to
\[
    \Delta E_{0,10}=E_{10}-E_0.
\]
For the finite system studied here, exact diagonalization gives $\Delta E_{0,10}\approx 4.3595$, as detailed in \cref{ap:Heisspectra}. This value is not the spectral gap $E_1-E_0$, but the target transition energy gap selected by our state preparation.

The time evolution is simulated for $N_t = 90$ time points with a step of $dt = 0.11$.
For both the TE-PAI and Trotter-based methods, the underlying Trotter decomposition uses $K_{\rm steps} = 650$ steps, and for TE-PAI, we set $\Delta = \pi/2^7$.
Although the simulation of the circuits is noise-free, i.e., free from gate or decoherence errors, we incorporate the statistical noise inherent in quantum measurements by taking only a single shot from each executed circuit.
To ensure a fair comparison, we fix the total number of circuit executions to 1000 at each time point.
For the standard Trotter-based method, this corresponds to simply taking $N_s = 1000$ shadow snapshots.
For the TE-PAI method, we set $M_{\text{TE-PAI}}\times N_s = 1000$ and explore configurations such as $(M_{\text{TE-PAI}}, N_s) = (1000, 1), (500, 2),$ and $(250, 4)$.
At each time point, we collect classical shadow snapshots and feed them into the algorithmic shadow spectroscopy protocol.
Within this post-processing, we estimate the expectation values of all 3-local Pauli operators from the snapshots and retain only the top 10\% of the most significant signals, which are identified using a Ljung-Box test (see also Appendix \ref{sec:shadow-spectroscopy} for detailed procedure). The Ljung–Box test and the retention of the top 10\% most significant signals constitute heuristic components of the spectral post-processing, adopted in line with prior work. A detailed assessment of the sensitivity to these choices is beyond the scope of this work; however, within the parameter ranges explored here, we do not observe qualitative changes in the extracted spectral features when these thresholds are varied moderately.

Figure \ref{fig_TE_PAI_shadow_spectro_same_circ_exec} shows the energy spectra obtained from the noise-free simulation.
As can be seen, both the Trotter-based method (blue dotted line) and our TE-PAI-based approach (solid lines) successfully identify the target transition energy gap at $\Delta E_{0,10} \approx 4.36$,  which validates the fundamental principle and correctness of our protocol.

A clear difference, however, appears in the signal intensity; for a fixed total number of 1000 circuit executions, the Trotter-based spectrum exhibits a significantly higher peak.
This reduced intensity for the TE-PAI method is an expected consequence of the additional statistical variance introduced by the quasiprobability sampling of gates.
As discussed in Ref.~\cite{Kiumi2025}, this variance necessitates a larger number of total samples to achieve the same signal-to-noise ratio as the deterministic Trotter evolution.
However, this higher sampling cost is the trade-off for the primary advantage of TE-PAI: a significant reduction in circuit depth, which provides crucial robustness against gate errors, as we will demonstrate later.

A second, notable finding from Fig. \ref{fig_TE_PAI_shadow_spectro_same_circ_exec} is the internal consistency of the TE-PAI results.
The peak intensity remains nearly identical regardless of how the 1000 total executions are distributed between the number of TE-PAI samples, $M_{\text{TE-PAI}}$, and the shadow snapshots per sample, $N_s$.
This is a non-trivial observation, as it empirically validates that the statistical precision of our hybrid protocol is determined solely by the total number of samples, not the specific strategy used to obtain them.
This finding has significant practical implications, as it offers the flexibility to choose the most convenient balance between generating many distinct evolution circuits (large $M_{\text{TE-PAI}}$) and performing more measurements on fewer circuits (large $N_s$) without affecting the final signal-to-noise ratio for a fixed total cost, see Appendix \ref{sec:Influence of the number of TE-PAI sample on the spectrum} for simulation with Large $N_s$.

Next, to assess the protocol's performance under more realistic conditions, we perform simulations on a 6-qubit Heisenberg model with a depolarizing noise model.
First, we construct the quantum circuits with a gate set consisting of {$R_X, R_Z, R_{XX}, R_{YY}, R_{ZZ}$} Pauli rotations.
We then introduce depolarizing noise after each gate, with an error probability of $p_1 = 10^{-4}$ for single-qubit gates and $p_2 = 10^{-3}$ for two-qubit gates.
For this experiment, we set the number of Trotter steps to $K_{\rm steps} = 300$ and the total number of circuit executions to 3000.
The PAI with $\Delta = \pi/2^5$ is used to reduce circuit depth.
In this setup, the Trotter-based circuits for $K_{\rm steps} = 300$ have an average depth of about 1810, while the corresponding TE-PAI circuits have an average depth of about 751. The depth is calculated as the minimum number of sequential layers of gates needed to execute the circuit, i.e., the length of the longest path through gates with qubits not sharing a layer. It was obtained in the simulation using Qiskit QuantumCircuit.depth method \cite{qiskit_depth}, version = "1.4.2". 

\begin{figure}
\includegraphics[width=0.48\textwidth]{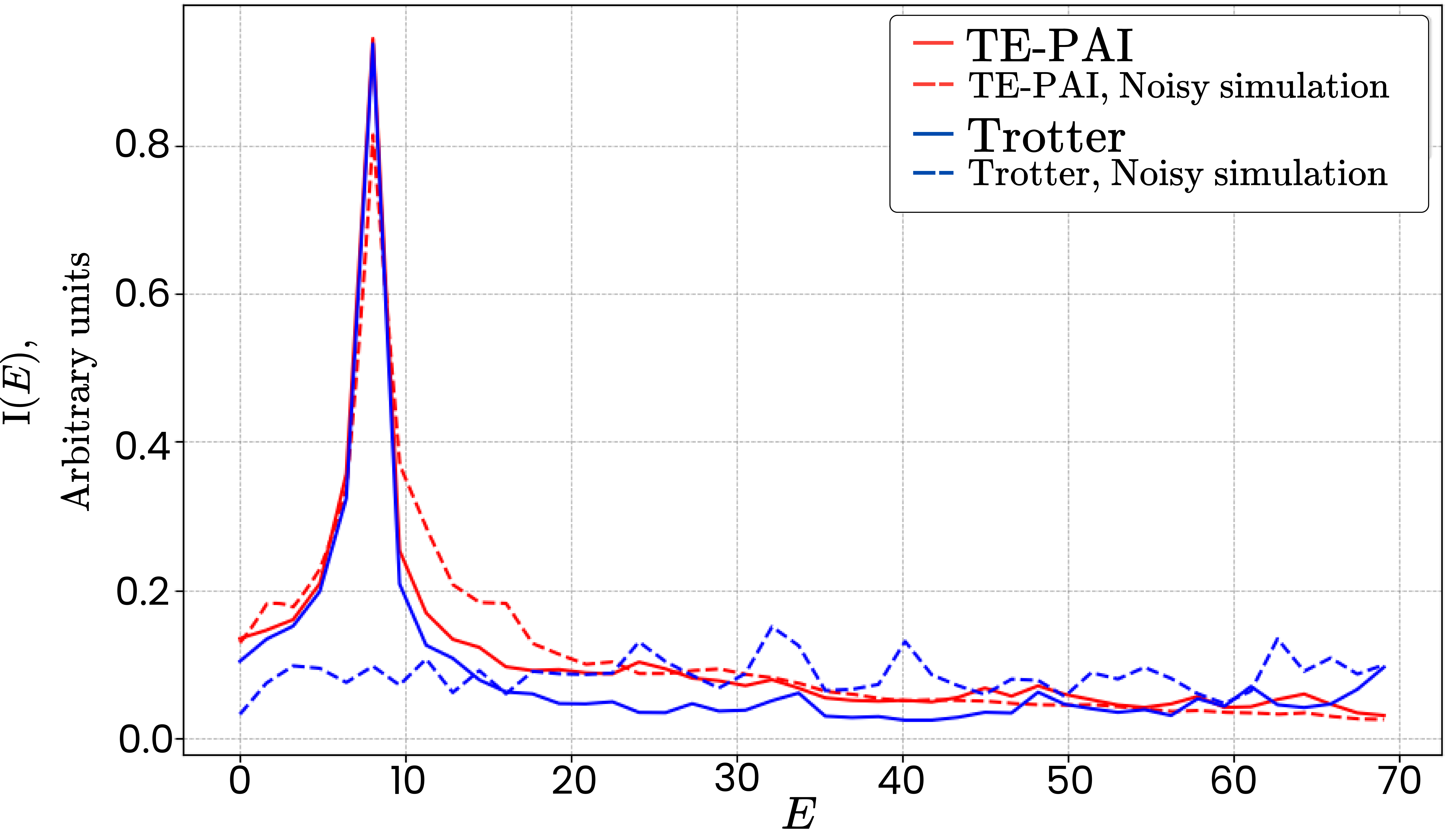}
    \caption{
        Effect of depolarizing noise on the energy spectra obtained from TE-PAI (red) and Trotter-based (blue) methods.
        Dotted lines represent the noise-free benchmarks, while solid lines show the results under noise.
        The TE-PAI spectrum demonstrates strong robustness, with its peak remaining clearly visible.
        In contrast, the Trotter-based spectrum collapses, losing nearly all spectral information.
        This highlights the practical advantage of TE-PAI's shallower circuits in noisy environments.
    }
\label{fig_TE_PAI_shadow_spectro_Heisenberg_noisy}
\end{figure}

Figure \ref{fig_TE_PAI_shadow_spectro_Heisenberg_noisy} shows the results of the noisy simulation, comparing them against the noise-free case.
The outcome reveals a stark difference in noise robustness.
The TE-PAI spectrum is only slightly degraded by the noise and the spectral peak remains clearly resolved.
In contrast, the Trotter-based spectrum collapses, as the noise almost entirely suppresses the peak.
This result confirms the superior noise resilience of the TE-PAI method, a direct consequence of its significantly shallower circuits.

\begin{figure*}
    \centering
    \begin{minipage}{0.49\textwidth}
        \centering
        \includegraphics[width=\textwidth]{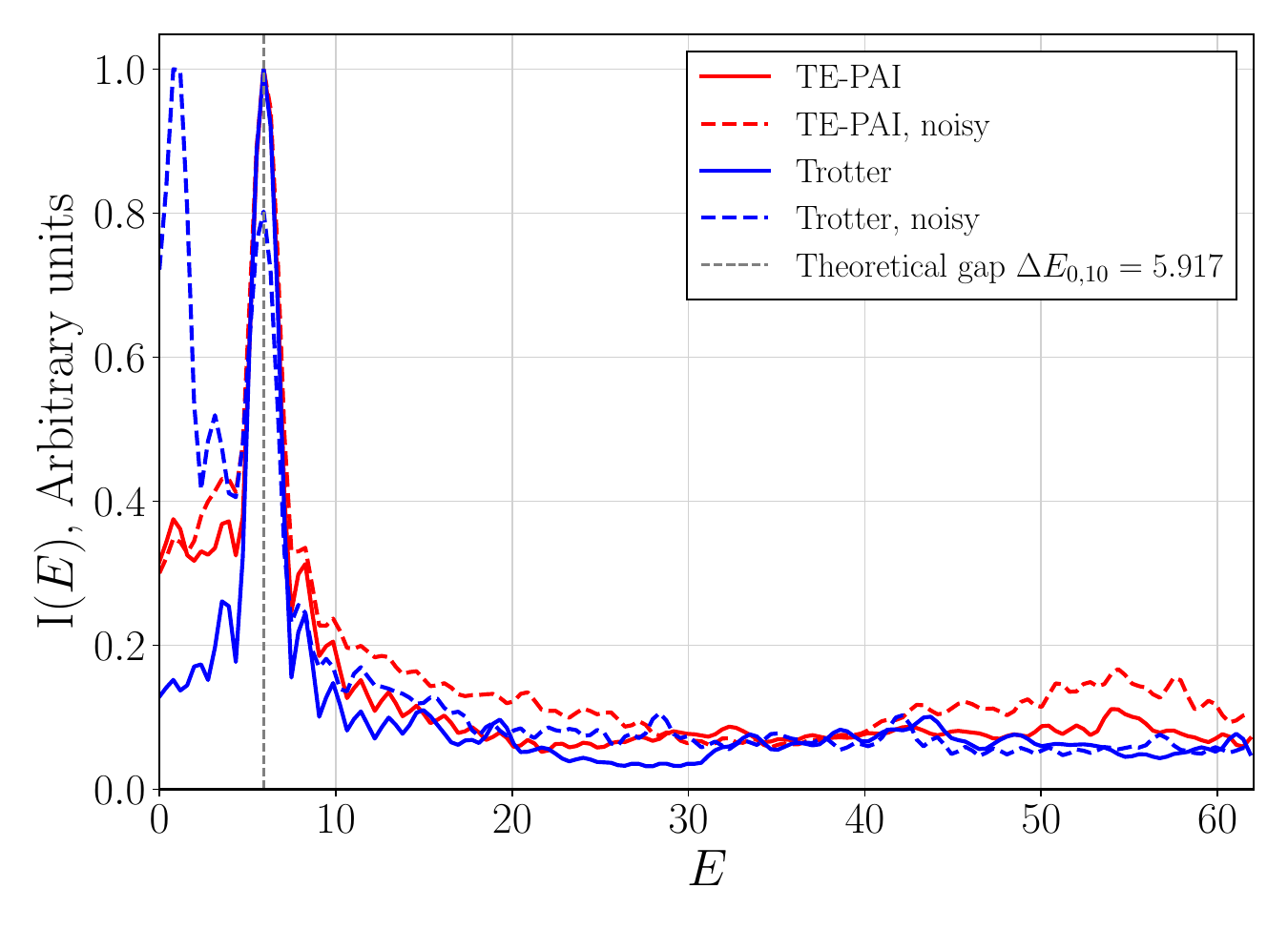}
        \\[-0.4em]
        (a) $M = 1000$ shots per time point.
    \end{minipage}
    \hfill
    \begin{minipage}{0.49\textwidth}
        \centering
        \includegraphics[width=\textwidth]{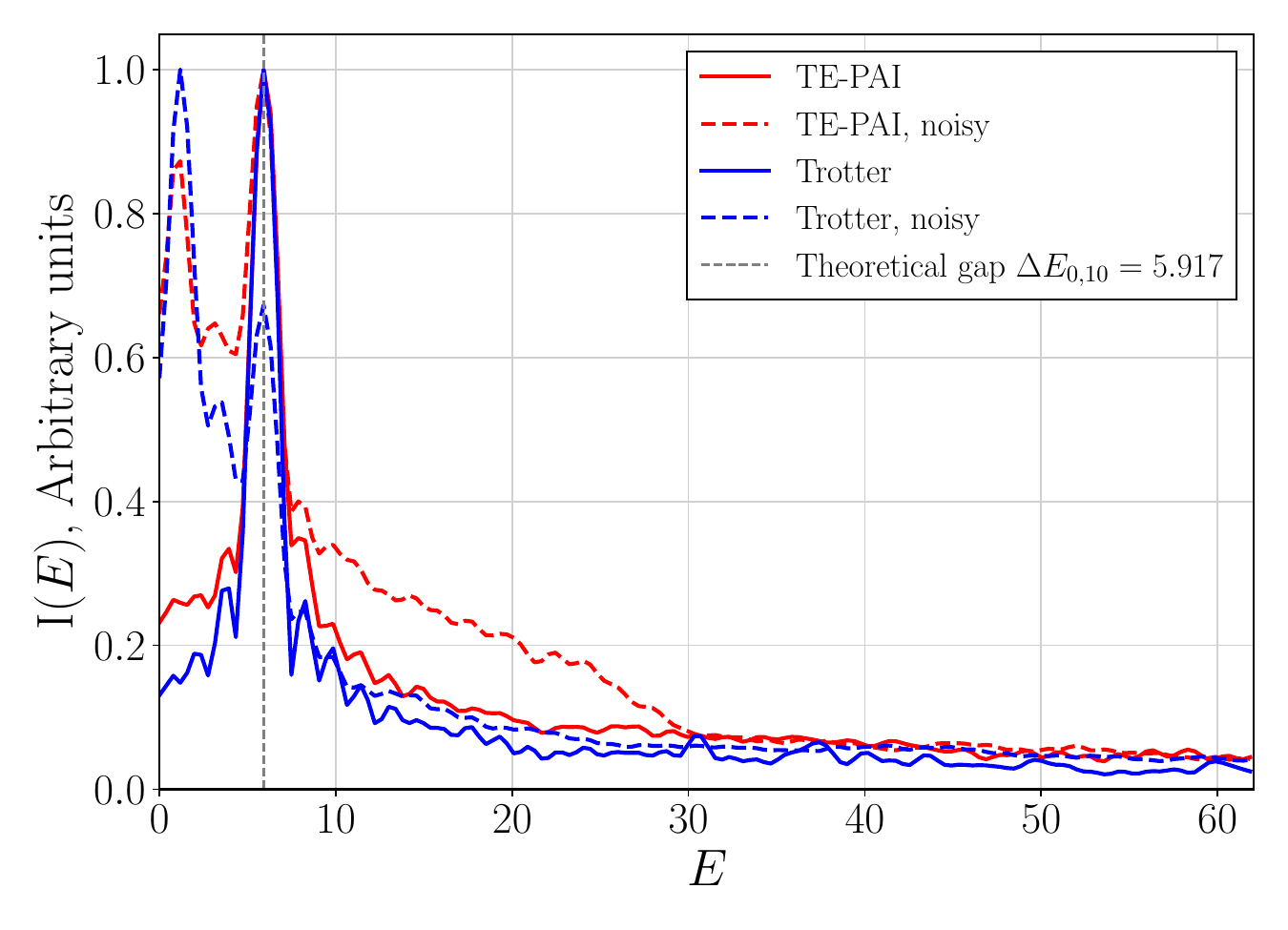}
        \\[-0.4em]
        (b) $M = 10000$ shots per time point.
    \end{minipage}
    \caption{
    Under \emph{amplitude-damping} ($T_1$) gate noise, a six-qubit Heisenberg chain ($J_x{=}J_y{=}J_z{=}1$, $\lVert H\rVert_1=15$) is prepared in $\ket{\psi_0}=\frac{1}{\sqrt{2}}\left(\ket{E_0}+\ket{E_{10}}\right)$ with target gap $\Delta E_{0,10}=5.917$ (grey dashed line). We compare TE-PAI (red) and standard Trotterization (blue), both noise-free (solid) and subject to amplitude damping (dashed; $p_1=2.5\times10^{-5}$, $p_2=2.5\times10^{-4}$). Both methods use the same constant step width $\delta t=0.008$ and $\Delta=\pi/2^{6}$, corresponding to the same $500$ time steps up to $t_{\max}=4$. They differ only in circuit depth: standard Trotterization applies $15$ gates per step, giving $15\times500=7500$ gates, whereas TE-PAI keeps only $2400$ gates on average—about $33\%$ as many—at the price of a quasiprobability sampling overhead $\Gamma\approx7.3$, which is absent for Trotterization. Increasing the number of shots from (a) $M=1000$ to (b) $M=10000$ suppresses
shot noise but does not necessarily improve the noisy spectrum. The deep Trotter
circuit fails in both panels, being dominated by a spurious low-frequency
decoherence peak. The shallower TE-PAI circuit keeps the correct gap peak
identifiable in both panels, but the larger shot budget also sharpens a
noise-induced low-frequency feature. Thus, additional sampling can make the
wrong spectral feature more visible rather than removing gate-noise-induced
errors.}
    \label{fig:both}
\end{figure*}

\begin{figure*}[t]
    \centering
\includegraphics[width=\textwidth]{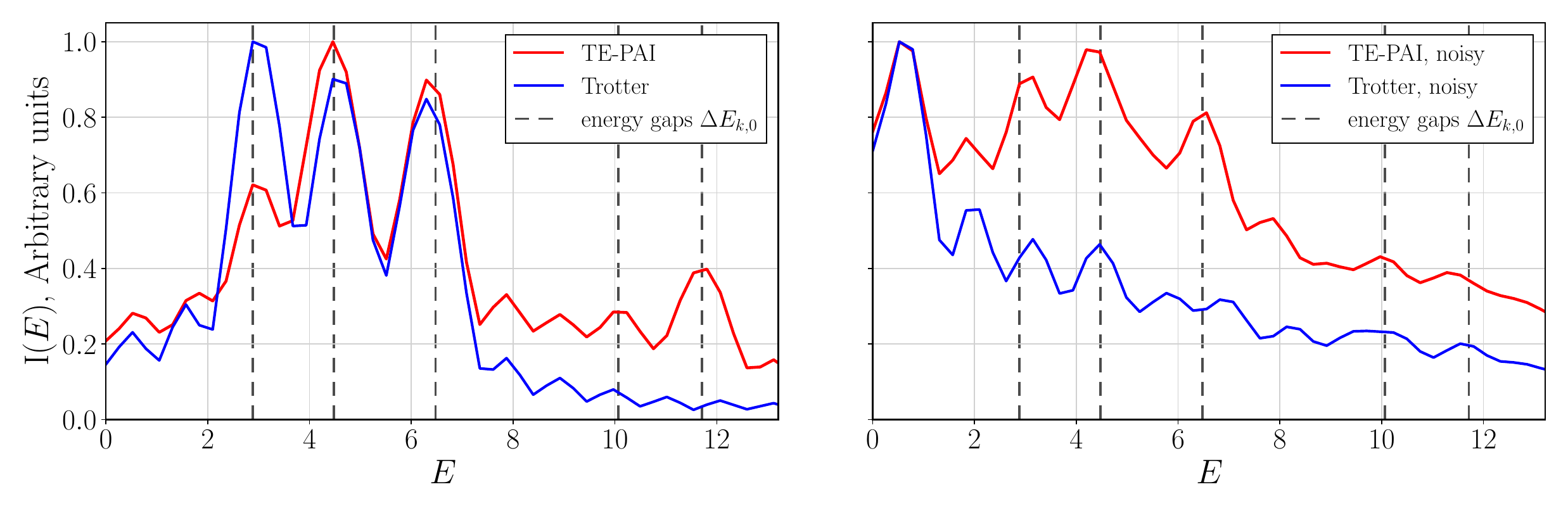}\\[-4mm]
    \caption{
    Robustness of TE-PAI shadow spectroscopy to gate noise for a \emph{multi-gap}
initial state, with $M=4000$ shots per time point. A five-qubit Heisenberg
chain ($J_x=J_y=J_z=1$, $\lVert H\rVert_1=12$) is prepared in a dominant-ground
superposition
$\ket{\psi_{\rm init}}
=
\frac{1}{\mathcal{N}}
\left(
\ket{E_0}
+
\epsilon\sum_k\ket{E_k}
\right)$
over the low-lying levels $\{1,2,4,7,9\}$, with $\epsilon=0.25$ and normalization factor $\mathcal{N}$.
The resulting dynamics contains several transition gaps $\Delta E_{k,0}$
(grey dashed lines), rather than a single frequency. The spectral intensity $I(E)$ is reconstructed using TE-PAI (red) and
Trotterization (blue). The left panel shows the noiseless spectra, while the
right panel includes depolarizing gate noise with
$p_1=4\times10^{-5}$ and $p_2=4\times10^{-4}$. Both methods use the same step width, with $\Delta=\pi/2^{6}$ and
$K_{\rm step}=700$ steps up to $t=6$. TE-PAI realizes this evolution with a $\sim\!3\times$ shallower
circuit, using $\sim\!2.9\times10^{3}$ gates on average instead of
$\sim\!8.4\times10^{3}$ Trotter gates, at a quasiprobability overhead
$\Gamma\approx10$. In the absence of gate noise, both methods resolve the low-lying transition
gaps. We consistently observe that TE-PAI gives higher-contrast peaks at several
higher-frequency gaps. Under gate noise, the shallower TE-PAI circuit preserves visible signatures
of the gap structure, especially in the higher-frequency region, whereas the
deeper Trotter circuit is strongly suppressed at the resolved gaps and drifts
toward a spurious low-frequency decoherence mode.}
    \label{fig:tepai_multigap}
\end{figure*}

To further probe the interplay between sampling and gate noise, in Figure~\ref{fig:both}, we apply TE-PAI
and standard Trotter shadow spectroscopy to a $6$-qubit Heisenberg chain
with $J_x=J_y=J_z=1$, $J=15$ Pauli terms, and $\lVert H\rVert_1=15$.
The initial state is prepared as
$\ket{\psi_{\rm init}}=\frac{1}{\sqrt{2}}
\left(
\ket{E_0}+\ket{E_{10}}
\right)$.
For this state, the ideal time signal contains a single dominant transition
frequency at $\Delta E_{0,10}=5.917$. The time-evolved state is reconstructed from $q=3$-local
Pauli shadows, and gate noise is modelled as amplitude damping, corresponding
to $T_1$ relaxation, with rates $p_1=2.5\times10^{-5}$ and
$p_2=2.5\times10^{-4}$ on one- and two-qubit gates, respectively. Both methods use the same first-order step size $\delta t=0.008$, i.e. the same
$K_{\rm step}=500$ steps up to $t_{\max}=4$ on a grid of $N_t=80$ time points, with the
TE-PAI angle $\Delta=\pi/2^{6}$.
The deterministic Trotter circuit applies
all $J$ Pauli rotations at every step, giving $7.5\times10^3$ gates.
In contrast, TE-PAI keeps only
$
\mathbb{E}[\nu_{\bm l}]\approx 2.4\times10^3
$
gates on average, about one third as many, at the cost of quasiprobability overhead $\Gamma\approx 7.3$.

For the Trotter circuit, the noisy limiting spectrum is already dominated by a
spurious low-frequency decoherence mode, and the reconstruction fails for both
shot budgets. Increasing $M$ therefore cannot restore the true transition gap;
it only reduces statistical fluctuations around an already distorted noisy
signal.

The TE-PAI reconstruction behaves differently. The correct gap peak remains
identifiable for both $M=1000$ and $M=10000$, confirming the benefit of the
shallower sampled circuits. However, increasing the shot budget does not
monotonically improve the spectrum. A larger $M$ reduces statistical fluctuations
and drives the estimator closer to the spectrum of the implemented noisy circuit
ensemble; under gate noise, this limiting spectrum is not the ideal Trotter
spectrum but the decohered spectrum generated by the noisy circuit. Consequently,
the noise-induced low-frequency feature can also become sharper, so the
reconstruction may look worse even though the statistical variance is reduced.
This suggests that the quasiprobability sampling noise, usually regarded as the
main drawback of TE-PAI, is not necessarily uniformly detrimental in this
spectroscopy task: at lower shot budgets, finite-sampling fluctuations may partly
obscure gate-noise-induced artefacts, while the true gap oscillation remains
visible due to the shallower circuits. This may also indicate that, in this instance, the genuine gap peak is relatively
robust against finite-sampling fluctuations.

Having established the protocol on a single energy gap, we next verify that its noise
robustness persists when several transition gaps must be resolved simultaneously. Multi-gap simulations are more stringent because the reconstructed spectrum must separate several weak transition peaks from cross-gap contributions and noise-induced artefacts. Following the standard practice for shadow spectroscopy~\cite{Chan2025AlgorithmicShadow},
we prepare a dominant-ground superposition
$\ket{\psi_{\rm init}}
=
\frac{1}{\mathcal{N}}
\left(
\ket{E_0}
+
\epsilon\sum_k\ket{E_k}
\right)$
over a set of low-lying eigenstates of a five-qubit Heisenberg chain
($J_x=J_y=J_z=1$, $\lVert H\rVert_1=12$), with $\epsilon=0.25$ and normalization factor $\mathcal{N}$.
For such a state, the spectral weights $c_0^{*}c_k\bra{E_0}O\ket{E_k}$ of the
ground-to-excited transitions dominate, so the time signal $S(t)$ in
\cref{eq:time_signal} carries several gap frequencies $\Delta E_{k,0}$ at once,
while the cross-gaps $\Delta E_{kl}$ with $k,l\neq0$ are suppressed as
$\mathcal{O}(\epsilon^{2})$. 

\Cref{fig:tepai_multigap} compares TE-PAI (red) and first-order Trotter (blue)
shadow spectroscopy for this state, without gate noise on the left and with
depolarizing gate noise on the right
($p_1=4\times10^{-5}$, $p_2=4\times10^{-4}$), using $q=3$-local Pauli shadows and
$M=4000$ shots per time point. The ground-dominant initial state is designed to
make the ground-to-excited transition gaps visible while keeping
excited-to-excited cross-gaps relatively suppressed.

Both methods use the same step width $\delta t$ ($\Delta=\pi/2^{6}$,
$K_{\rm steps}=700$ steps up to $t=6$) and therefore incur the same Trotter
error; the spectra are evaluated on $N_t=120$ time points. The two differ in
circuit depth: TE-PAI realises this accuracy with a $\sim\!3\times$ shallower
circuit ($\sim\!2.9\times10^{3}$ versus $\sim\!8.4\times10^{3}$ gates), at the
cost of a sampling overhead
$\Gamma\approx10$. In the noiseless
case both resolve the low-lying gaps, but with a complementary bias: standard
Trotter captures the three lowest transitions
($\Delta E_{k,0}=2.88,\,4.48,\,6.48$) most strongly while leaving the two
highest ($10.07,\,11.71$) at the floor, whereas TE-PAI under-weights the
lowest transition yet retains visible peaks at the two highest gaps, where it
outperforms Trotter. 

Under gate noise, this difference is decisive: the deeper Trotter circuit
accumulates more decoherence and its spectrum is strongly suppressed at the
resolved gaps, drifting toward a spurious low-frequency mode. In contrast, the
shallower TE-PAI circuit preserves visible signatures of the gap structure, with
higher contrast at several resolved gaps, especially in the higher-frequency
region. The advantage of TE-PAI shadow spectroscopy is therefore not specific to
the single-gap setting but carries over to the simultaneous estimation of
multiple transition gaps.

\subsection{Experiments on quantum hardware}

To demonstrate the practical viability of our protocol, we present experiments for a 20-qubit transverse-field Ising model on IBM quantum hardware.
The Hamiltonian for this model is
\begin{align}
    H = -J \sum_{k=1}^{N-1}Z_kZ_{k+1} - d\sum_{k=1}^{N}X_k.
\end{align}
The theoretical spectral gap between the ground and first excited states is $\Delta E = 2\sqrt{J^2 + d^2 - 2Jd\cos(\pi/(N+1))}$ \cite{pfeuty1970onedimensional}.
As a practical compromise between preparation efficiency and sufficient overlap with these states, we use the initial state $|{+}\cdots{+}0\rangle$ and set the parameters to $J=0.1$ and $d=2$. With this choice of initial state and the subsequent shadow-spectroscopy filtering, the dominant observed spectral feature corresponds to the ground--first-excited-state transition, while other transitions have much weaker spectral weight in the processed signal.
We set the number of Trotter steps to $K_{\rm steps}=115$ and use a TE-PAI angle of $\Delta = \pi/2^5$, which limits the average circuit depth to approximately 600.
For the time evolution, we use $N_t=80$ steps with an interval of $dt=0.037$.
At each time point, we perform 3000 total circuit executions, configured as $(M_{\text{TE-PAI}}, N_s) = (3000, 1)$.
From the resulting shadow snapshots, the algorithmic shadow spectroscopy process is performed in exactly the same manner as the classical simulation.
We run these experiments on the \textit{ibm\_kobe} and \textit{ibm\_kingston} quantum computers, whose detailed properties are provided in Appendix \ref{app:Hardware properties}.

Figure \ref{fig_Hardware simulation} shows the spectra from our experiments on quantum hardware.
The data from both the \textit{ibm\_kobe} (red line) and \textit{ibm\_kingston} (green line) devices correctly resolve the theoretical energy gap, confirming our protocol's viability under realistic noise conditions.
We observe a higher peak from \textit{ibm\_kobe}, which is consistent with its lower average gate error rate as detailed in Appendix \ref{app:Hardware properties}.
Crucially, the TE-PAI spectrum is more pronounced than the one from the conventional Trotter-based method (blue dotted line) run on the same device.
This result provides direct experimental evidence that the shallower circuits produced by TE-PAI lead to superior robustness against hardware noise.
It thereby marks a significant scale-up from the 6-qubit hardware demonstration in the original work \cite{Chan2025AlgorithmicShadow}.

\begin{figure}
    \includegraphics[width=0.5\textwidth]{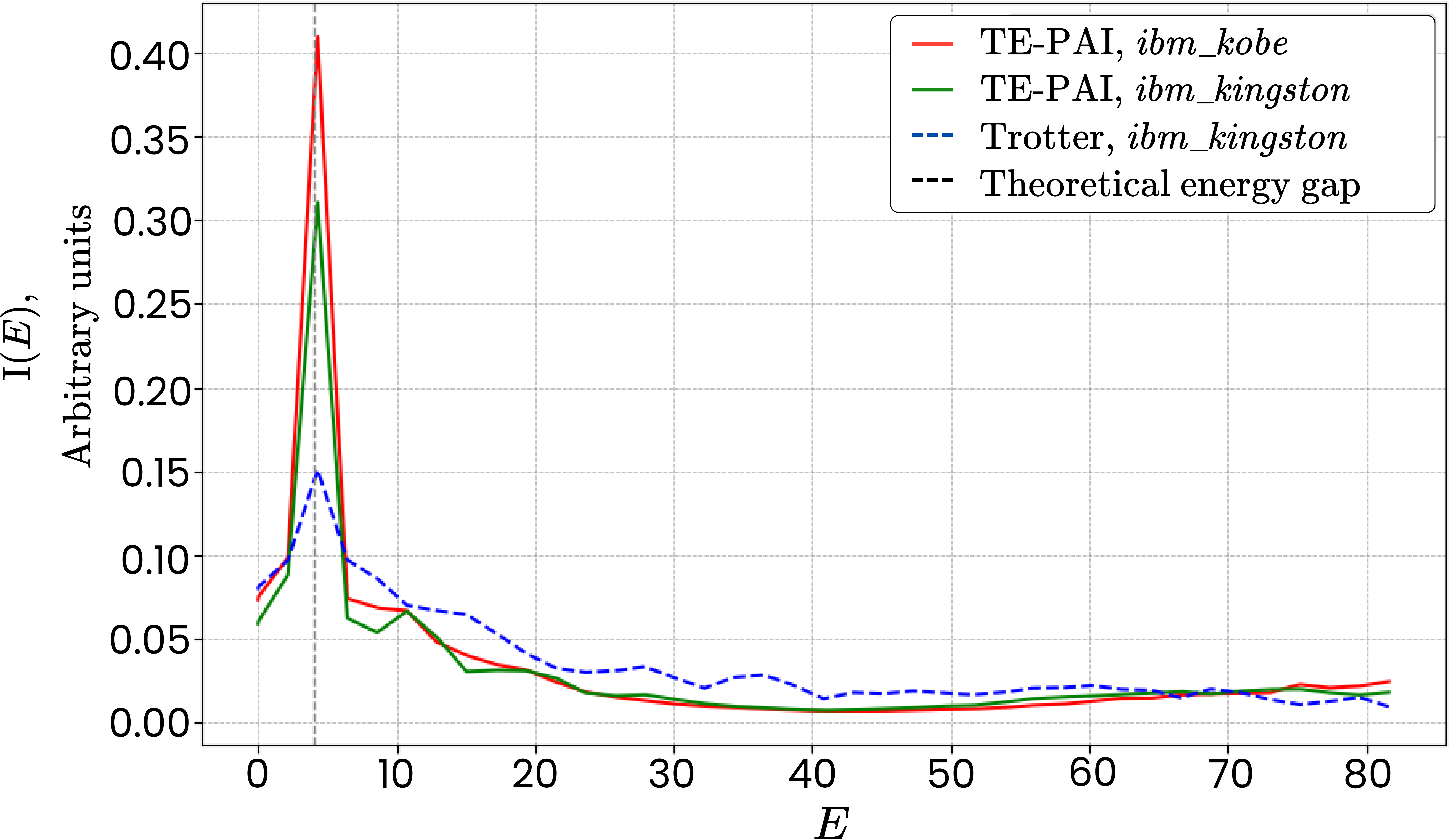}
    \caption{Spectra of a 20-qubit Transverse Ising Hamiltonian measured on quantum hardware. The red and green lines show a TE-PAI-based spectrum performed respectively on \textit{ibm\_kobe} and \textit{ibm\_kingston}. The blue dotted line shows a Trotter-based spectrum performed on \textit{ibm\_kingston}. The grey dotted vertical line indicates the exact energy gap.}
    \label{fig_Hardware simulation}
\end{figure}

\section{Conclusion}\label{sec:conclusion}

In this work, we introduced TE-PAI shadow spectroscopy, a hybrid quantum-classical protocol that integrates the randomized time evolution of TE-PAI with the measurement framework of algorithmic shadow spectroscopy.
We have theoretically established that this combination yields an unbiased estimator for spectral properties and experimentally validated its performance.
Our noise-free classical simulations confirmed that the protocol correctly identifies energy gaps, while also highlighting an inherent trade-off: the quasiprobability sampling in TE-PAI introduces statistical variance that requires a larger number of total circuit executions to achieve the same signal intensity as the first-order Trotter-based method. However, this cost is made up for by the protocol's primary advantage, which we demonstrated in noisy simulations: a significant reduction in circuit depth leads to superior robustness against gate errors.
We also found that, under gate noise, a moderate shot budget can preserve the correct peak while suppressing spurious noise-induced features. In contrast, increasing the shot budget can sharpen these spurious features and thereby worsen the spectral estimate. This suggests that the sampling overhead of TE-PAI is not merely a drawback in this task: by avoiding over-resolution of gate-noise-induced artifacts, it can sometimes help robust energy-gap estimation. Also, our additional multi-gap simulations further showed that TE-PAI can resolve several transition gaps and tends to identify higher-frequency gaps more clearly than the Trotter-based approach.
This key benefit was further confirmed through experiments on 20-qubit quantum hardware, where our method successfully resolved the target energy gap and outperformed the conventional approach, demonstrating its practical viability for spectroscopy in a significantly larger scale experiment than previously performed.

Let us briefly address the practical time complexity of our protocol compared to standard approaches. While TE-PAI utilizes $M_{\text{TE-PAI}}$ unique circuit instances rather than a single reused circuit, the efficiency of modern control stacks prevents this from causing a significant increase in wall-clock execution time. Current quantum runtime environments (such as Qiskit Runtime) support circuit batching, allowing the full ensemble of $M_{\text{TE-PAI}}$ circuits to be submitted and executed within a single job, thereby avoiding the latency overhead associated with multiple job submissions. Furthermore, our variance analysis (\cref{thm:tepai_shadow_variance}) and numerical results (\cref{fig_TE_PAI_shadow_spectro_same_circ_exec}) demonstrate that the estimator remains unbiased even when reusing a single compiled circuit instance for $N_s$ shots. This allows for the compilation and loading costs to be amortized over $N_s$ measurements. By optimizing the ratio of unique circuits ($M$) to shots-per-circuit ($N_s$) for a fixed total budget, one can balance the statistical requirements against the hardware's specific compilation and loading latencies, rendering the total wall-clock time comparable to standard Trotterized evolution.

Looking forward, TE-PAI shadow spectroscopy opens several promising research avenues.
The natural trade-off between circuit depth and sampling overhead constitutes a tunable design parameter that can be optimized in accordance with the specific performance characteristics of the underlying quantum hardware. In particular, quantum devices with lower error rates can reliably support deeper circuits, thereby reducing the required sampling overhead. Conversely, on hardware with higher error rates, circuit depth can be deliberately curtailed at the expense of increased sampling overhead, enabling a more favorable balance between noise resilience and statistical accuracy.
An important direction for future work is to quantify how much different product formulas, including symmetric second-order and randomized Trotter formulas~\cite{Childs2019Randomization}, can improve the performance of TE-PAI shadow spectroscopy on quantum hardware.
A natural next step would be to combine our protocol with other error mitigation techniques, such as readout error correction or zero-noise extrapolation, to further suppress the impact of device imperfections.
Furthermore, the demonstrated scalability and noise resilience make our method a strong candidate for tackling challenging problems in quantum chemistry and condensed matter physics, where deep circuits have been a major bottleneck.
The flexibility in distributing the sampling budget between the number of TE-PAI circuits and shadow snapshots also suggests possibilities for hardware-aware co-design of experiments, potentially leading to even more efficient spectral analysis on near-term quantum computers.

\noindent\textbf{Data availability}: The simulation code used in this work is available online at: \url{https://github.com/hugopgs/TE-PAI_shadow_spectroscopy}.

\section*{Acknowledgments}
Y.M. is supported by JST SPRING under Grant Number JPMJSP2138. This work is partially supported by JST ASPIRE Japan Grant Number JPMJAP2319, the project JPNP20017 funded by the New Energy and Industrial Technology Development Organization (NEDO), JST COI-NEXT program Grant No. JPMJPF201.
K.M. is supported by JST FOREST Grant No. JPMJFR232Z, JSPS KAKENHI Grant No. 23H03819, 24K16980 and JST CREST Grant No. JPMJCR24I4.
C.K. is supported by JST PRESTO, Japan, Grant Number JPMJPR25F1, JSPS KAKENHI, Grant Number JP26K17052. This work, as part of the ITI 2021–2028 program of the University of Strasbourg, CNRS and Inserm, was supported by IdEx Unistra (ANR-10-IDEX-0002), and by SFRI STRAT’US project (ANR-20-SFRI-0012) and EUR QMAT ANR-17-EURE-0024 under the framework of the French Investments for the Future Program.
B.K. thanks UKRI for the Future Leaders Fellowship Theory to Enable Practical Quantum Advantage (MR/Y015843/1)
and EPSRC for funding through the project Software Enabling Early Quantum Advantage (SEEQA, EP/Y004655/1).
This research was funded in part by UKRI (MR/Y015843/1).
For the purpose of Open Access, the author has applied a CC BY public copyright licence
to any Author Accepted Manuscript version arising from this submission.

\appendix
\crefalias{section}{appendix}

\section{Product Formulas}

Throughout this work, we employ the first-order Trotter product formula. The associated Trotter error, arising from the discretization of time evolution, is an intensively studied topic \cite{Childs_2021}; here, we briefly summarize the relevant results. Let $U_{H,t}^{(K)}$ denote the Trotter circuit defined in \cref{eq:trotter_circuit} with $K$ Trotter steps. The spectral norm error between the exact time evolution operator $U_{H,t}$ and the approximation $U_{H,t}^{(K)}$ is bounded by:
\begin{equation*}\left\| U_{H,t} - U_{H,t}^{(K)} \right\|\le \frac{t^2\|c\|_T}{2K},\end{equation*}

where $\|c\|_T := \sum_{\gamma<\beta}\left\|[h_\gamma P_\gamma,h_\beta P_\beta]\right\|$ is the commutator norm of $H$. The error in the estimated expectation value induced by the Trotter approximation is bounded by the operator error and the spectral norm of the observable:
\begin{align*}
    \epsilon_T
    &:=\left| S(t) - \langle O \rangle^{(K)}_{t} \right|\\
    &\le 2\|O\|\,\left\| U_{H,t} - U_{H,t}^{(K)} \right\| \\
    &\le \frac{t^2 \|O\|\|c\|_T }{K},
\end{align*}
where $\langle O \rangle^{(K)}_{t} := \bbraket{O|\bm{U}_{H,t}^{(K)}(\rho)}$ denotes the expectation value obtained by evolving $\rho$ with the $K$-step Trotter circuit. As $K$ increases, the Trotter approximation converges to the exact evolution, so the Trotter error vanishes in the limit $K\to\infty$ (the bound above scales as $O(t^2/K)$); moreover, this is a worst-case upper bound, and the empirical error for specific instances is often much smaller \cite{Childs_2021}.

Other than the Trotter error, finite sampling introduces statistical error. To estimate $\langle O\rangle^{(K)}_{t}$ on hardware, we sample $N$ independent measurement outcomes $o_1,\dots,o_N$ from repeated executions of the Trotter circuit $U_{H,t}^{(K)}$ and use the sample mean $\hat{o}:=\frac{1}{N}\sum_{n=1}^N o_n$.
Then $\mathbb{E}\!\left[\hat{o}\right]=\langle O\rangle^{(K)}_{t}$ and $\Var\!\left(\hat{o}\right)\le \|O\|^2/N,$
leading to a sample complexity scaling as
\begin{equation*}
  N = \mathcal{O}\!\left( \frac{\|O\|^2\log(1/\delta)}{\epsilon^2} \right),
  \label{equ:M_sample}
\end{equation*}
to achieve additive error \(\epsilon\) with failure probability \(\delta\).

\section{Quasiprobability sampling}\label{app:quasiprobability}
The theoretical results presented in \cref{sec:formal_analysis} are not limited to TE-PAI but apply generally to any quasi-probability sampling method. Here, we briefly review this underlying framework.

In this approach, quantum channels that are difficult to implement exactly are approximated by randomly sampling from a set of simpler operations. This strategy is mathematically grounded in quasi-probability decompositions, which serve as the foundation for many randomized compiling techniques. Prominent examples include probabilistic error cancellation (PEC) for stochastic noise mitigation \cite{Temme2017ErrorMitigation,Endo2018PracticalQEM,Jnane2024QuantumErrorMitigated}, probabilistic angle interpolation (PAI) for coherent or synthesis-error mitigation \cite{Koczor2024ProbabilisticInterpolation,Koczor_2024,Kiumi2025}, and circuit cutting \cite{cutting1,cutting2,cutting3}.

Suppose a target channel $\bm{C}$ can be decomposed as a linear combination 
\[\bm{C}=\sum_{d=1}^{D} a_d \bm{C}_d,\]
 where $a_d \in \mathbb{R}$ are real coefficients and $\{\bm{C}_d\}$ is a set of efficiently implementable basis channels. We define the normalization factor $\gamma(\bm{C})=\sum_{d=1}^{D} |a_d|$ and assign sampling probabilities $p_d=|a_d|/\gamma(\bm{C})$.

To simulate $\bm{C}$, we sample an index $l \in \{1,\dots, D\}$ from the distribution $\{p_l\}$, apply the channel $\bm{C}_l$, and weight the measurement outcome by $\gamma(\bm{C})\mathrm{sign}(a_l)$. This weighting is performed via classical post-processing; specifically, for any observable $O$ and state $\rho$, the expectation value satisfies:
\begin{align*}
    \bbraket{O|\bm{C}(\rho)}
    = \gamma(\bm{C})
    \mathbb{E}\left[
    \mathrm{sign}(a_l)
    \bbraket{O|\bm{C}_l(\rho)}
    \right],
\end{align*}
where the expectation is taken over the random choice of $l$. Consequently, the resulting random estimator is unbiased, satisfying
 \[\mathbb{E}[\gamma(\bm{C})\mathrm{sign}(a_l)\bm{C}_l]=\bm{C}.\]

 Next, we consider a sequential channel $\bm{C}=\bm{C}_K\cdots \bm{C}_1$ with each $\bm{C}_k$ decomposed as 
 \[\bm{C}_k=\sum_{d=1}^D a_{k,d}\bm{C}_{k,d}\] 
 where $\bm{C}_{k,d}$ are efficiently implementable channels.  
Let $l_k\in \{1,\ldots,D\}$ be a random variable drawn according to $p_{d}=|a_{k,d}|/\gamma(\bm{C}_k)$, where $\gamma(\bm{C}_k)=\sum_d |a_{k,d}|$. 
Define the random channel 
\begin{align*}
    \hat{\bm{C}}_{\bm{l}} :=\bm{C}_{K,l_K} \cdots \bm{C}_{2,l_2}\bm{C}_{1,l_1}
\end{align*}
with the classical post-processing weight $\Gamma_{\bm{l}} = \prod_{k=1}^K \gamma(\bm{C}_{k})\mathrm{sign}(a_{k,l_k})$. For any initial state $\rho$, applying the sampled channel sequence produces an unbiased estimator of $\bm{C}(\rho)$ as  
\[\mathbb{E}[\Gamma_{\bm{l}}\hat{\bm{C}}_{\bm{l}}(\rho)] = \bm{C}(\rho).\]
While this re-weighting guarantees unbiasedness, it scales the estimator's variance by approximately $\Gamma^2:=\prod_{k=1}^K \gamma(\bm{C}_{k})^2$, thereby requiring the number of samples to increase by a factor of $\Gamma^2$ to maintain a fixed estimation error.

\section{Higher-order TE-PAI: gate count and sampling overhead}
\label{app:higher-order-te-pai}

We compare $K$ first-order steps with $K/2$ symmetric second-order steps over the same total time $t$, where $K$ is even and $\delta t=t/K$:
\begin{align*}
 U_{H,t}^{(1,K)}
 &=\left[\prod_{j=1}^{J}e^{-ih_jP_j\delta t}\right]^K,\\
 U_{H,t}^{(2,K/2)}
 &=\left[
 \prod_{j=1}^{J}e^{-ih_jP_j\delta t}
 \prod_{j=J}^{1}e^{-ih_jP_j\delta t}
 \right]^{K/2}.
\end{align*}
The second-order step has duration $2\delta t$, with each sweep covering its half-step. Both circuits contain $K$ copies of each rotation $R_{P_j,\theta_j}$, where $\theta_j=2h_j\delta t$, and hence $KJ$ rotations in total. Their global product-formula errors scale as $O(t\delta t)$ and $O(t\delta t^2)$, respectively, up to Hamiltonian-dependent constants~\cite{Suzuki1,Suzuki2,Childs_2021}.

For a TE-PAI decomposition of a rotation angle $\theta$, the sampling overhead is given by
\[
    \gamma(\theta):=|a_1(\theta)|+|a_2(\theta)|+|a_3(\theta)|,
\]
where $a_1,a_2,a_3$ are defined in Eq.~\eqref{eq:te-pai-decomp-prelim-app}. For $|\theta|\leq\Delta$,
\begin{equation*}
    \gamma(\theta)
    =
    \cos|\theta|
    +
    \tan\!\left(\frac{\Delta}{2}\right)\sin|\theta|,
\end{equation*}
and hence

\begin{equation}
 \log\gamma(\theta)
 =|\theta|\tan\!\left(\frac{\Delta}{2}\right)+O(\theta^2).
 \label{eq:log-gamma-expansion-app}
\end{equation}

The probability that a rotation of angle $\theta$ is sampled as a non-identity gate is
\begin{align*}
 p_{\neq I}(\theta)
 &=\frac{|a_2(\theta)|+|a_3(\theta)|}{\gamma(\theta)}\\
 &=\frac{3-\cos\Delta}{2\sin\Delta}|\theta|+O(\theta^2).
\end{align*}

Let $\nu_{\bm l}$ denote the number of non-identity gates in a sampled TE-PAI circuit.

Applying TE-PAI independently to these rotations with the same $\Delta\geq\max_j|\theta_j|$ gives exactly the same expected number of sampled non-identity rotations and quasiprobability overhead:
\begin{align*}
 \mathbb{E}[\nu^{(1,K)}]
 &=\mathbb{E}[\nu^{(2,K/2)}]
 =K\sum_{j=1}^{J}p_{\neq I}(\theta_j),\\
 \Gamma_{1,K}^2
 &=\Gamma_{2,K/2}^2
 =\left[\prod_{j=1}^{J}\gamma(\theta_j)\right]^{2K}.
\end{align*}
Thus, second order improves the Trotter-error scaling without increasing either of these TE-PAI costs. These identities do not imply equal actual estimator variances.

The comparison above is made before merging adjacent rotations. For example, the two central rotations in a symmetric second-order step satisfy $R_{P_J,\theta_J}R_{P_J,\theta_J}=R_{P_J,2\theta_J}$, reducing two deterministic rotations to one. If they are merged before applying TE-PAI, their contribution to the expected rotation count changes from $2p_{\neq I}(\theta_J)$ to $p_{\neq I}(2\theta_J)$, provided $2|\theta_J|\leq\Delta$. Thus, merging reduces the deterministic rotation count, but its effect on the TE-PAI cost depends on the merged angle and the choice of $\Delta$.

More generally, if a product formula is written as
\begin{equation*}
    U_{H,t}^{(p,K)}
    =
    \prod_{r=1}^{L_{p,K}} R_{P_r,\theta_r},
\end{equation*}
then the leading TE-PAI costs are governed by the total absolute rotation angle:
\begin{equation*}
    \log \Gamma_p^2
    =
    2\tan\!\left(\frac{\Delta}{2}\right)
    \sum_{r=1}^{L_{p,K}}|\theta_r|
    +
    O\!\left(
        \sum_{r=1}^{L_{p,K}}\theta_r^2
    \right),
\end{equation*}
and
\begin{equation*}
    \mathbb{E}\!\left[\nu_{\bm l}^{(p)}\right]
    =
    \frac{3-\cos\Delta}{2\sin\Delta}
    \sum_{r=1}^{L_{p,K}}|\theta_r|
    +
    O\!\left(
        \sum_{r=1}^{L_{p,K}}\theta_r^2
    \right).
\end{equation*}
 In the comparison above, both formulas have total absolute rotation angle $K\sum_j|\theta_j|=2t\|H\|_1$. In contrast, higher-order Suzuki formulas of order four and above can contain negative or enlarged fractional time steps, which increase the total absolute rotation angle and can therefore increase the TE-PAI gate count and sampling overhead. We note that a continuous-time version of TE-PAI has recently been studied in Ref.~\cite{dai2026structure}, where the limiting ensemble can be interpreted as a Dyson-series-based sampling scheme with a natural leading $l_1$ rate. From this viewpoint, fourth- and higher-order Suzuki formulas are not expected to converge to this optimal leading cost, because their negative or enlarged fractional time steps increase the total absolute rotation angle.

\section{Algorithmic shadow spectroscopy}\label{sec:shadow-spectroscopy}

Let $H$ be an $N$-qubit Hamiltonian with eigenstates $\ket{E_\alpha}$ for $\alpha = 1, \dots, 2^N$ and corresponding eigenvalues $E_\alpha$.
To obtain energy differences, $E_\alpha - E_\beta$, algorithmic shadow spectroscopy uses expectation values of a set of observables $\{O_i\}_{i=1}^{N_o}$ consisting of all $q$-local Pauli operators, evaluated with respect to time-evolved states $\kket{\rho(t)} = \bm{U}_{H,t} \kket{\rho_{\rm init}}$ for some initial state $\kket{\rho_{\rm init}}$ \cite{Chan2025AlgorithmicShadow}.

The algorithm is as follows.
Let $\{t_n\}_{n=1}^{N_t}$ be a set of time points and define $S_i(n) := \bbraket{O_i | \rho(t_n)}$.
\begin{enumerate}
    \item Gather shadow snapshots of $\rho(t_n)$ for each $t_n$ using $N_s$ copies of $\rho(t_n)$. This is the only step performed on quantum hardware; the following steps are classical.

    \item Estimate $\hat{S}_i(n)$ from the snapshots.

    \item Construct a data matrix $D \in \mathbb{R}^{N_o \times N_t}$ with entries
          $D_{i, n}  = \frac{\hat{S}_i(n) - \mu_i}{\sigma_i}$
          and $\mu_i$, $\sigma_i$ are the mean and standard deviation of $\{\hat{S}_i(n)\}_{n=1}^{N_t}$, respectively.

    \item Apply the Ljung-Box test to each row of $D$ to remove signals indistinguishable from noise \cite{Ljung_box_test1, Ljung_box_test2}.

    \item Compute the correlation matrix
          \begin{equation}
              C = D D^\top \in \mathbb{R}^{N_o \times N_o},
          \end{equation}
          and extract the top $c$ dominant eigenvectors $\{v_1, v_2, \dots, v_c\}$ corresponding to the largest eigenvalues \cite{BRONSON2021185}.

    \item Compute the spectral cross-correlation matrix
          \begin{equation}
              x_{kl}(m) = \sum_{n=1}^{N_t - m - 1} v_k(n + m) \, v_l(n),
          \end{equation}
          and perform a discrete Fourier transform (DFT) of $x_{kl}(m)$ with respect to the time lag index $m$. Denote the resulting quantity by $X_{kl}(m)$, where $m$ now represents the frequency index.

    \item Compute the largest singular value $\lambda(m)$ of the matrix $X(m)$. Peaks in $\lambda(m)$ correspond to the transition energy gaps of the Hamiltonian.
\end{enumerate}

\section{Formal Analysis of the TE-PAI shadow Estimator}\label{app:experimental_algorithm}

Algorithm~\ref{alg:te-pai-shadow-experimental} shows the detailed pseudocode of the practical implementation of our protocol  used in the experiments. 
As discussed in the main text, this implementation is designed to generate one time-evolution circuit via the TE-PAI process and reuses it with $N_s$ different randomized measurement bases to obtain $N_s$ snapshots for each TE-PAI sampled circuit. \Cref{thm:tepai_unbiased_shadow,thm:tepai_shadow_variance} establish that this modified implementation yields an unbiased estimator and provides an explicit upper bound on the variance.

Before proving \cref{thm:tepai_unbiased_shadow,thm:tepai_shadow_variance}, we first present a lemma that upper bounds the variance of an expectation value for a single quasiprobability random circuit.

\begin{lemma}\label{lem:tepai_single_variance}
    The expectation value of any observable $O$ under a quasiprobability random circuit yields an unbiased estimator of the true expectation value:
     \[\mathbb{E}[\Gamma_{\bm{l}}\bbraket{O|\hat{\bm{C}}_{\bm{l}} (\rho)}]=\bbraket{O|\bm{C}(\rho)}.\]
     Additionally, the variance is upper bounded as follows:
\begin{align*}
\operatorname{Var}\left[ \Gamma_{\bm{l}}\bbraket{O|\hat{\bm{C}}_{\bm{l}} (\rho)}\right]
\leq \Gamma ^{2}\| O\|^{2}.
\end{align*}

\end{lemma}
\begin{proof}
    The unbiasedness directly follows from the unbiasedness of the TE-PAI decomposition:
    \begin{align*}
    \mathbb{E}[\Gamma_{\bm{l}}\bbraket{O|\hat{\bm{C}}_{\bm{l}} (\rho)}]
    &=
    \bbraket{O|\mathbb{E}[\Gamma_{\bm{l}}\hat{\bm{C}}_{\bm{l}} ](\rho )}
    \\
    &=\bbraket{O|\bm{C}(\rho)}
\end{align*}
    Write $\hat{o} := \Gamma_{\bm{l}}\bbraket{O|\hat{\bm{C}}_{\bm{l}} (\rho)}$ for brevity. The variance is given by
$
\operatorname{Var} [\hat{o} ] =\mathbb{E}\left[\hat{o}^{2}\right] -\mathbb{E}[\hat{o}]^{2} =\mathbb{E}\left[\hat{o}^{2}\right] -\bbraket{O|\bm{C}(\rho)}^{2},
$
and we now upper bound the first term as
\begin{equation*}
\begin{aligned}
\mathbb{E}\left[\hat{o}^{2}\right] & =\sum _{\bm{l}} p_{\bm{l}} \times \left[\Gamma_{\bm{l}}\bbraket{O|\bm{C}_{\bm{l}} (\rho)}\right]^{2}\\
 & =\Gamma ^{2}\sum _{\bm{l}} p_{\bm{l}} \times \bbraket{O|\bm{C}_{\bm{l}} (\rho )}^{2}\\
 & \leq \Gamma ^{2} \| O\|^{2}
\end{aligned}
\end{equation*}
Above we have used that $|\bbraket{O|\rho }|\leq \| O\|$ for any density matrix $\rho$ and $\sum _{\bm{l}} p_{\bm{l}} =1$.
\end{proof}

\subsection{Proof of Theorem~\ref{thm:tepai_unbiased_shadow}}\label{app:proof_of_unbiasedness}
This modified implementation preserves the unbiased nature of the estimator.
\begin{proof}
    The goal is to estimate the true expectation value of an observable $O$, defined with respect to the original channel
    $\bbraket{O|\bm{C}(\rho)}.$ Our final estimator, $\hat{\langle O\rangle}$, is the average over all measurement outcomes.
    More specifically, let $o_{m,s}$ be the outcome from the $m$-th quasiprobability random circuit and the $s$-th shadow snapshot, that is,
    $$
        o_{m,s} = \Gamma_{\bm{l}_m} \bbraket{O | \hat{\rho}_{m,s}},
    $$
    where $\Gamma_{\bm{l}_m}$ is the weight accumulated from the $m$-th quasiprobability sample and $\hat{\rho}_{m,s}$ is the corresponding classical shadow snapshot defined in \eqref{eq:shadow-snapshot-prelim}.
    The estimator is constructed as
    $$
        \hat{\langle O\rangle} = \frac{1}{M N_s} \sum_{m=1}^{M} \sum_{s=1}^{N_s} o_{m,s}.
    $$
    The expectation value of this estimator, $\mathbb{E}[\hat{\langle O \rangle}]$, is taken over both the probabilistic choice of $\bm{l}_m$ and the shadow's random choice of measurement bases. We can analyze the expectation in steps, from the inside out using the {\it law of total expectation}.
    First, the inner expectation is a conditional expectation over the shadow measurement outcomes, given a fixed quasiprobability random circuit $\bm{C}_{\bm{l}_m}$. We denote this expectation by $\mathbb{E}_{\text{sh}}$ and we have
    \begin{align}
        \mathbb{E}_{\text{sh}} \left[ \frac{1}{N_s} \sum_{s=1}^{N_s} o_{m,s} \right] 
        &=\frac{1}{N_s} \sum_{s=1}^{N_s}\mathbb{E}_{\text{sh}} \left[  o_{m,s} \right] \nonumber
        \\
        \label{eq:inner_expectation}
        &= \bbraket{O|\Gamma_{\bm{l}_m} \bm{C}_{\bm{l}_m}(\rho)}.
    \end{align}
    Next, we consider the outer average over the $M$ samples.
    The Quasiprobability sampling protocol is constructed to be an unbiased estimator of the exact time-evolution channel, meaning
    $\mathbb{E}_{\bm{l}}\left[ \Gamma_{\bm{l}_m} \bm{C}_{\bm{l}_m} \right] = \bm{C},$
    By the linearity of expectation and trace, the full expectation value is
    \begin{align*}
        \mathbb{E}[\hat{\langle O \rangle}]
        &=\mathbb{E}_{\bm{l}}\left[\frac{1}{M}\sum_{m=1}^{M}\mathbb{E}_{\text{sh}} \left[ \frac{1}{N_s} \sum_{s=1}^{N_s} o_{m,s}  \right]\right]\\
         & = \frac{1}{M}\sum_{m=1}^{M}\mathbb{E}_{\bm{l}} \left[\bbraket{O|\Gamma_{\bm{l}_m} \bm{C}_{\bm{l}_m}(\rho)}\right] \\
         & = \bbraket{O|\mathbb{E}_{\bm{l}}[\Gamma_{\bm{l}_m} \bm{C}_{\bm{l}_m}](\rho)}   \\
         & = \bbraket{O|\bm{C}(\rho)}   \\
    \end{align*}
\end{proof}

\subsection{Proof of Theorem~\ref{thm:tepai_shadow_variance}}\label{app:proof_of_variance}

\begin{proof}
Let $Y_m$ denote the estimator obtained from the $m$-th sampled circuit, averaged over its $N_s$ shadow snapshots:
\begin{equation*}
    Y_m = \frac{1}{N_s} \sum_{s=1}^{N_s} o_{m,s}.
\end{equation*}
Since the $M$ circuit samples are independent and identically distributed, the total variance is $\text{Var}[\langle \hat{O} \rangle] = \frac{1}{M} \text{Var}[Y_m]$. We apply the Law of Total Variance to $Y_m$, conditioning on the choice of the random circuit $\bm{C}_{\bm{l}_m}$. Let $\mathbb{E}_{\bm{l}}$ and $\text{Var}_{\bm{l}}$ denote the expectation and variance over the circuit distribution given by the quasiprobability sampling, and $\mathbb{E}_{\text{sh}}$ and $\text{Var}_{\text{sh}}$ denote the conditional expectation and variance over the shadow measurement randomness, given a fixed random circuit $\bm{C}_{\bm{l}_m}$. The total variance decomposes as
\begin{equation*}
    \text{Var}[Y_m] = \text{Var}_{\bm{l}} \left( \mathbb{E}_{\text{sh}}[Y_m] \right) + \mathbb{E}_{\bm{l}} \left[ \text{Var}_{\text{sh}}(Y_m ) \right].
\end{equation*}
Conditioned on a specific circuit $\bm{C}_{m}$, the shadow estimator is unbiased for that circuit's observable value as \cref{eq:inner_expectation}. From \cref{lem:tepai_single_variance}, the first term is upper bounded as
\begin{align} \label{eq:var_circuit}
    \text{Var}_{\bm{l}} \left(\mathbb{E}_{\text{sh}}[Y_m] \right) \le \Gamma ^{2}\| O\|^{2}-\bbraket{O|\bm{C}(\rho)}^{2}.
\end{align}
Conditioned on a fixed random circuit $\mathcal{C}_{\bm{l}_m}$, $Y_m$ is the average of $N_s$ independent snapshots. The variance of a single weighted snapshot $\Gamma_m \text{Tr}(O \hat{\rho})$ is determined by the second moment of the shadow estimator as $\mathbb{E}_{\text{sh}} [ \bbraket{O|\hat{\rho}}^2] \le  \|O\|_{\text{sh}}^2.$ Thus, the variance of a single snapshot given $\mathcal{C}_m$ is:
\begin{align*}
    \text{Var}_{\text{sh}}(\Gamma_{\bm{l}_m} \bbraket{O|\hat{\rho}_{m,s}}) 
    &\le \Gamma^2 (\|O\|_{\text{sh}}^2 - \bbraket{O|\bm{C}_{\bm{l}_m}(\rho)}^2).
\end{align*}
Since $Y_m$ averages $N_s$ such snapshots, its conditional variance is scaled by $1/N_s$:
\begin{equation*}
    \text{Var}_{\text{sh}}(Y_m) = \frac{\Gamma^2}{N_s}  (\|O\|_{\text{sh}}^2 - \bbraket{O|\bm{C}_{\bm{l}_m}(\rho)}^2).
\end{equation*}
Taking the expectation over circuits yields the second term of the total variance:
\begin{align} \label{eq:var_shadow}
    \mathbb{E}_{\bm{l}} \left[ \text{Var}_{\text{sh}}(Y_m) \right]
    & = \frac{\Gamma^2}{N_s} \left(  \|O\|_{\text{sh}}^2 - \mathbb{E}_{\bm{l}}[\bbraket{O|\bm{C}_{\bm{l}_m}(\rho)}^2] \right).\nonumber\\
    & \le \frac{\Gamma^2}{N_s}\left( \|O\|_{\text{sh}}^2 - \|O\|^2 \right) .
\end{align}
Summing Eq.~\eqref{eq:var_circuit} and Eq.~\eqref{eq:var_shadow} gives $\text{Var}[Y_m]$:

\begin{equation*}
    \text{Var}[Y_m] \le \frac{\Gamma^2}{N_s}  \|O\|_{\text{sh}}^2 + \left( 1 - \frac{1}{N_s} \right) \Gamma^2\|O\|^2.
\end{equation*}
Finally, dividing by $M$ yields the variance of the total estimator $\langle \hat{O} \rangle$:
\begin{align*}
    \operatorname{Var}[\langle \hat{O} \rangle] &\leq \Gamma ^{2}\left(\frac{\| O\|^2 _{\text{sh}}-\| O\|^2}{M N_s} +\frac{\| O\|^2}{M}\right).
\end{align*}
\end{proof}

\begin{figure}[t]
\begin{algorithm}[H]
\caption{Experimental implementation of shadow snapshots
\label{alg:te-pai-shadow-experimental}}
\begin{algorithmic}[1]
    \Statex \textbf{Input:} Hamiltonian $H$, initial state $\rho_{\text{init}}$, total time $t$, Trotter steps $K_{\rm steps}$, angle $\Delta$, number of TE-PAI samples $M_{\text{TE-PAI}}$, number of shadow snapshots per sample $N_s$
    \Statex \textbf{Output:} A collection of $N=M_{\text{TE-PAI}} \times N_s$ shadow snapshot tuples $(\Gamma, \{U_n\}_{n=1}^N, \{b_n\}_{n=1}^N)$

    \State Initialize an empty list of snapshots $\mathcal{S}$
    \For{$m = 1$ to $M_{\text{TE-PAI}}$}
        \State Initialize TE-PAI weight $\Gamma_m \leftarrow 1$
        \State Initialize an empty quantum circuit for evolution $\mathcal{C}_{\text{evol}}$
        \State Set $\delta = t / K_{\rm steps}$
        \For{$k = 1$ to $K_{\rm steps}$}
            \For{$j=1$ to $J$}
                \State Compute $\theta_j= 2h_j\delta$ and coefficients $a_1, a_2, a_3$
                \State Set $\gamma = |a_1|+|a_2|+|a_3|$.
                \State Sample $d$ with probability $|a_d|/\gamma$
                \State Update $\Gamma_m \leftarrow \Gamma_m \times \gamma \times \mathrm{sign}(a_d)$
                \State Append the corresponding gate to $\mathcal{C}_{\text{evol}}$
            \EndFor
        \EndFor
        \For{$s = 1$ to $N_s$}
            \State Create a copy: $\mathcal{C}_{\text{total}} \leftarrow \mathcal{C}_{\text{evol}}$
            \For{$n = 1$ to $N$}
                \State Sample a Clifford gate $U_n$ uniformly at random
                \State Append $U_n$ to qubit $n$ in $\mathcal{C}_{\text{total}}$
            \EndFor
            \State Apply $\mathcal{C}_{\text{total}}$ to $\rho_{\rm init}$ and measure to get bitstring $b$
            \State Append the tuple $(\Gamma_m, \{U_n\}, b)$ to  $\mathcal{S}$
        \EndFor
    \EndFor
    \State \Return $\mathcal{S}$
\end{algorithmic}
\end{algorithm}
\end{figure}

\section{Energy Spectrum of the Heisenberg Hamiltonian} \label{ap:Heisspectra}
In condensed matter physics, the isotropic Heisenberg model is typically defined using spin-1/2 operators $\mathbf{S}_i = \frac{1}{2}\boldsymbol{\sigma}_i$, where $\boldsymbol{\sigma} = (X, Y, Z)$ are the Pauli matrices \cite{takahashi1999, mattis1981}. The standard Hamiltonian is:\begin{equation}\mathcal{H} = J \sum_{i=1}^{N-1} (\mathbf{S}i \cdot \mathbf{S}{i+1}).\end{equation}Our simulation employs Pauli operators directly, implying a scaling of $\boldsymbol{\sigma}_i \cdot \boldsymbol{\sigma}_{i+1} = 4(\mathbf{S}_i \cdot \mathbf{S}_{i+1})$. Consequently, all eigenvalues are scaled by a factor of 4 relative to the classic Bethe Ansatz results \cite{bethe1931}. For an infinite chain, the ground state energy $E_0$ follows the Hulthén integral \cite{hulthen1938}, while low-lying excitations are described by des Cloizeaux-Pearson magnon triplets \cite{descloizeaux1962}. In our finite $N=10$ system with open boundary conditions and $J=1$, exact diagonalization reveals a discrete spectrum governed by $SU(2)$ symmetry. The ground state and the 10th energy level, which resides in the third excited manifold, yield the following numerical results: 
$E_{Ground state}= -17.0321$ and $E_{10^{th}}= -12.6726$ 
resulting in the target transition energy gap for our state preparation, $\Delta E_{0,10}=E_{10}-E_0\approx 4.3595$.
This gap represents the energy required to sustain multi-magnon excitations within the quantized momentum space of a 10-qubit chain.
We emphasize that this is not an adjacent-level spacing. It is the transition energy gap selected by preparing a superposition of the ground state and the 10th energy level, and therefore it appears as the dominant frequency in the time-dependent signal.

\section{Influence of the Shadow size on the TE-PAI spectrum}
\label{sec:Influence of the number of TE-PAI sample on the spectrum}
We now examine how shadow size influences the resulting spectrum. Fig. \ref{fig_TE_PAI_shadow_spectro_Different_Ns} presents three TE-PAI spectra obtained using a fixed number of TE-PAI samples across varying shadow sizes, $N_s \in \{1, 10, 100\}$. All simulations utilized the parameters defined in Sec. \ref{sec:classical_simulation}. While all three spectra successfully resolve the target transition energy gap, the peak intensity increases significantly with larger shadow sizes. However, a larger shadow size imposes a substantial computational overhead, as the total number of circuit executions scales with $M_{\text{TE-PAI}} \times N_s$.
In this study, we fix the shadow size to $N_s = 1$  and only change the number of TE-PAI samples.

\begin{figure}
    \centering
    \includegraphics[width=0.5\textwidth]{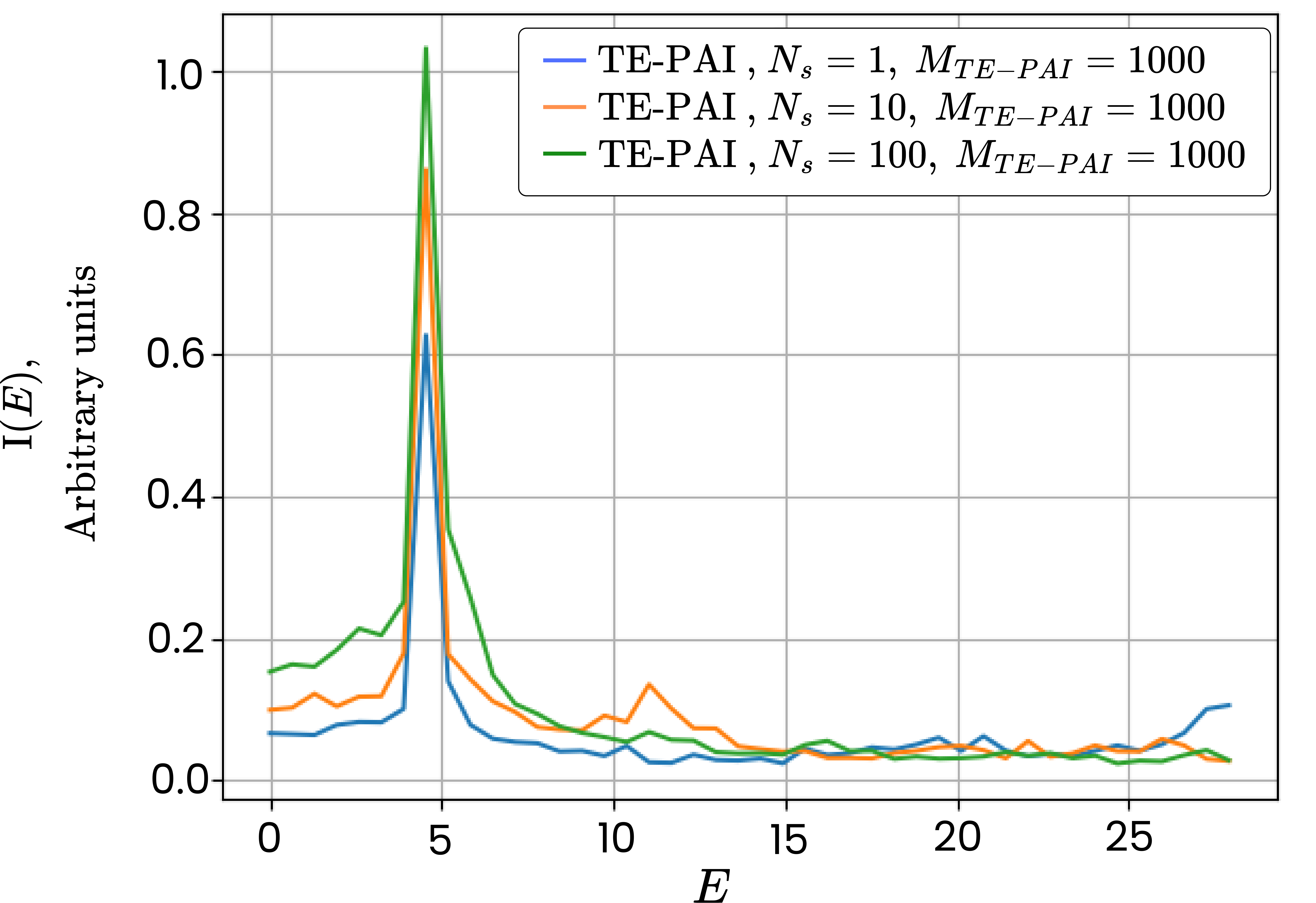}
    \caption{Comparison of TE-PAI spectra computed with different shadow sizes. The spectra were generated using a fixed number of TE-PAI samples while varying the shadow size $N_s \in \{1, 10, 100\}$, corresponding to the blue, orange, and green spectra, respectively. While all cases resolve the target transition energy gap, increasing $N_s$ significantly enhances the peak intensity. Simulation parameters match those detailed in Sec. \ref{sec:classical_simulation}.}
    \label{fig_TE_PAI_shadow_spectro_Different_Ns}
\end{figure}

\section{Hardware properties}
\label{app:Hardware properties}

\Cref{Tab:backendproperties} summarizes the hardware characteristics of the two IBM Quantum backends used in this study: \textit{ibm\_kobe} and \textit{ibm\_kingston}. Both devices are based on the Heron r2 architecture and provide identical native basis gate sets, consisting of \texttt{cz}, \texttt{id}, \texttt{rx}, \texttt{rz}, \texttt{rzz}, \texttt{sx}, and \texttt{x}. This extensive gate set allows for greater flexibility in circuit design, particularly by enabling simulation of more complex Hamiltonians natively, without requiring excessive transpilation. Despite architectural similarities, the backends differ slightly in their two-qubit gate fidelity. \textit{ibm\_kobe} exhibits a lower two-qubit gate error, both in the best-case scenario ($1.03\times 10^{-3}$) and in layered execution ($2.68\times 10^{-3}$), compared to \textit{ibm\_kingston}, which shows $1.06\times 10^{-3}$ and $3.64\times 10^{-3}$, respectively. This improved fidelity makes \textit{ibm\_kobe} a better candidate for experiments requiring two-qubit gates. Although the differences in fidelity are marginal, \textit{ibm\_kobe}'s slightly better two-qubit performance and its alignment with low transpilation depth makes it our preferred choice for hardware-based simulations involving entanglement-heavy operations.

\begin{table}[t]
\centering
\caption{Comparison of processor performance metrics between \textit{ibm\_kobe} and \textit{ibm\_kingston}.}
\label{Tab:backendproperties}
\begin{ruledtabular}
\begin{tabular}{lcc}
Processor & \textit{ibm\_kobe} & \textit{ibm\_kingston} \\
2Q error (best)         & $1.03\times10^{-3}$ & $1.06\times10^{-3}$ \\
2Q error (layered)      & $2.68\times10^{-3}$ & $3.64\times10^{-3}$ \\
CLOPS                  & $2.5\times10^{5}$  & $2.5\times10^{5}$  \\
Median CZ error         & $1.858\times10^{-3}$ & $2.088\times10^{-3}$ \\
Median SX error         & $2.383\times10^{-4}$ & $2.379\times10^{-4}$ \\
\end{tabular}
\end{ruledtabular}
\end{table}

The logical-to-physical mapping is performed based on an optimization strategy prioritizing long $T_1$/$T_2$ and low error rates. 
% In particular, qubits with poor readout fidelity or short lifetimes were excluded. 
% \cref{fig_qubits kobe} and \cref{fig_qubits kingston} show the physical mapping and the characteristics of each qubit used during the hardware simulation on IBM quantum machine with respectively  \texttt{ibm\_kobe} and \texttt{ibm\_kingston}. 

\Cref{Mapping_kobe} and \Cref{Mapping_kingston} show diagram of the physical mapping of qubits.

% \begin{figure}
%     % \centering
%     \includegraphics[width=0.45\textwidth]{figures/ibm_kobe_properties_qubits.pdf}
%     \caption{Mapping and properties of \texttt{ibm\_kobe} qubits used during the hardware simulation}
%     \label{fig_qubits kobe}
% \end{figure}

% \begin{figure}
%     % \centering
%     \includegraphics[width=0.45\textwidth]{figures/Tab_Qubits_kingston.pdf}
%     \caption{Mapping and properties of \texttt{ibm\_kingston} qubits used during the hardware simulation}
%     \label{fig_qubits kingston}
% \end{figure}

\begin{figure}[!h]
    % \centering
    \includegraphics[width=0.45\textwidth]{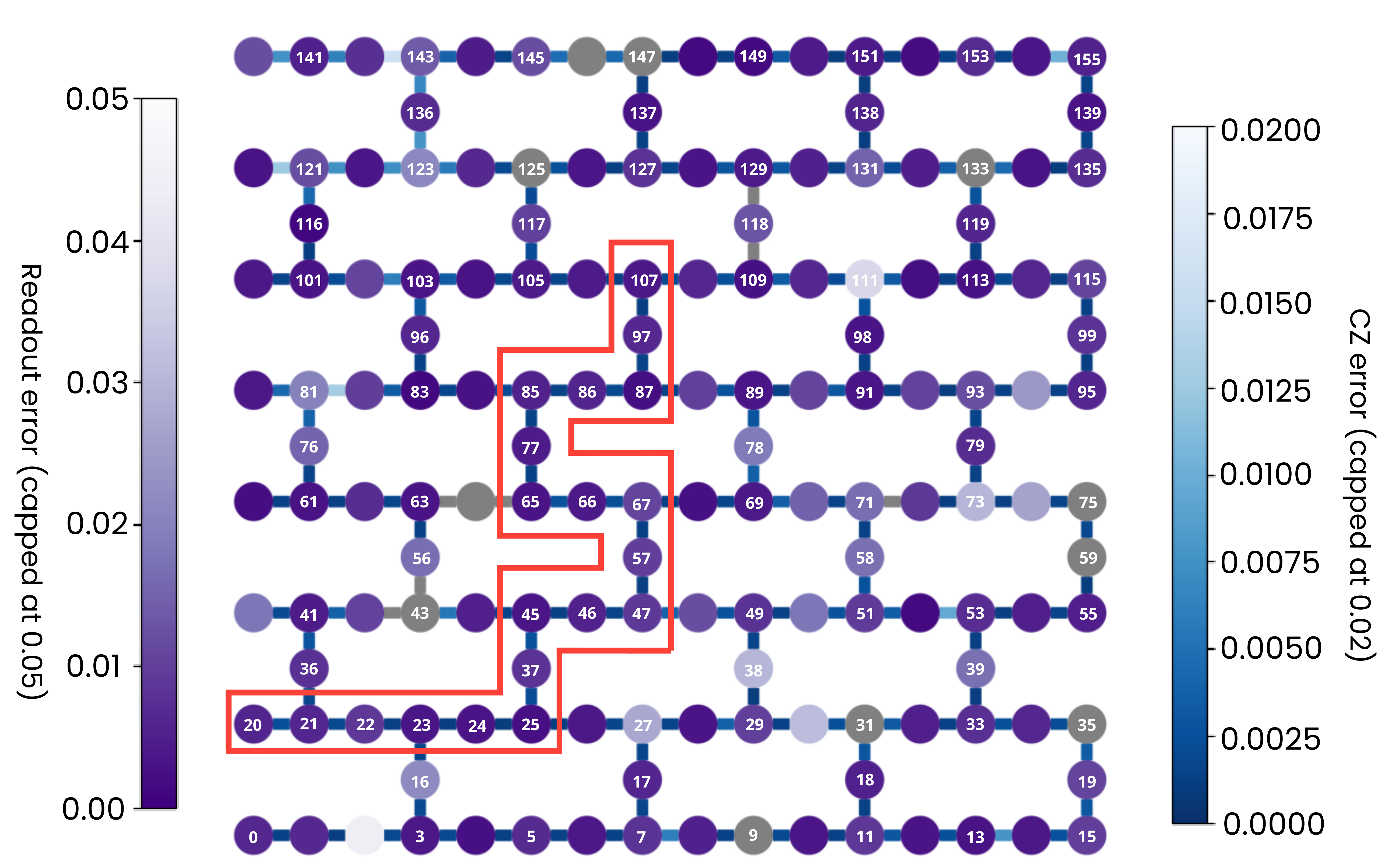}
    \caption{Architecture and error landscape of the \textit{ibm\_kobe} quantum processor.
Each node represents a qubit, with its color indicating the average readout assignment error. Each edge represents a physical connection between two qubits, and its color encodes the average controlled-Z (CZ) gate error. Grey nodes or edges indicate that the corresponding error exceeds a predefined threshold, set to 0.05 for the readout error and 0.02 for the CZ error.
The qubits and connections highlighted by the red outline denote the specific subset of the device that is utilized in our experiments.
}
    \label{Mapping_kobe}
\end{figure}

\begin{figure}[!h]
    % \centering
    \includegraphics[width=0.45\textwidth]{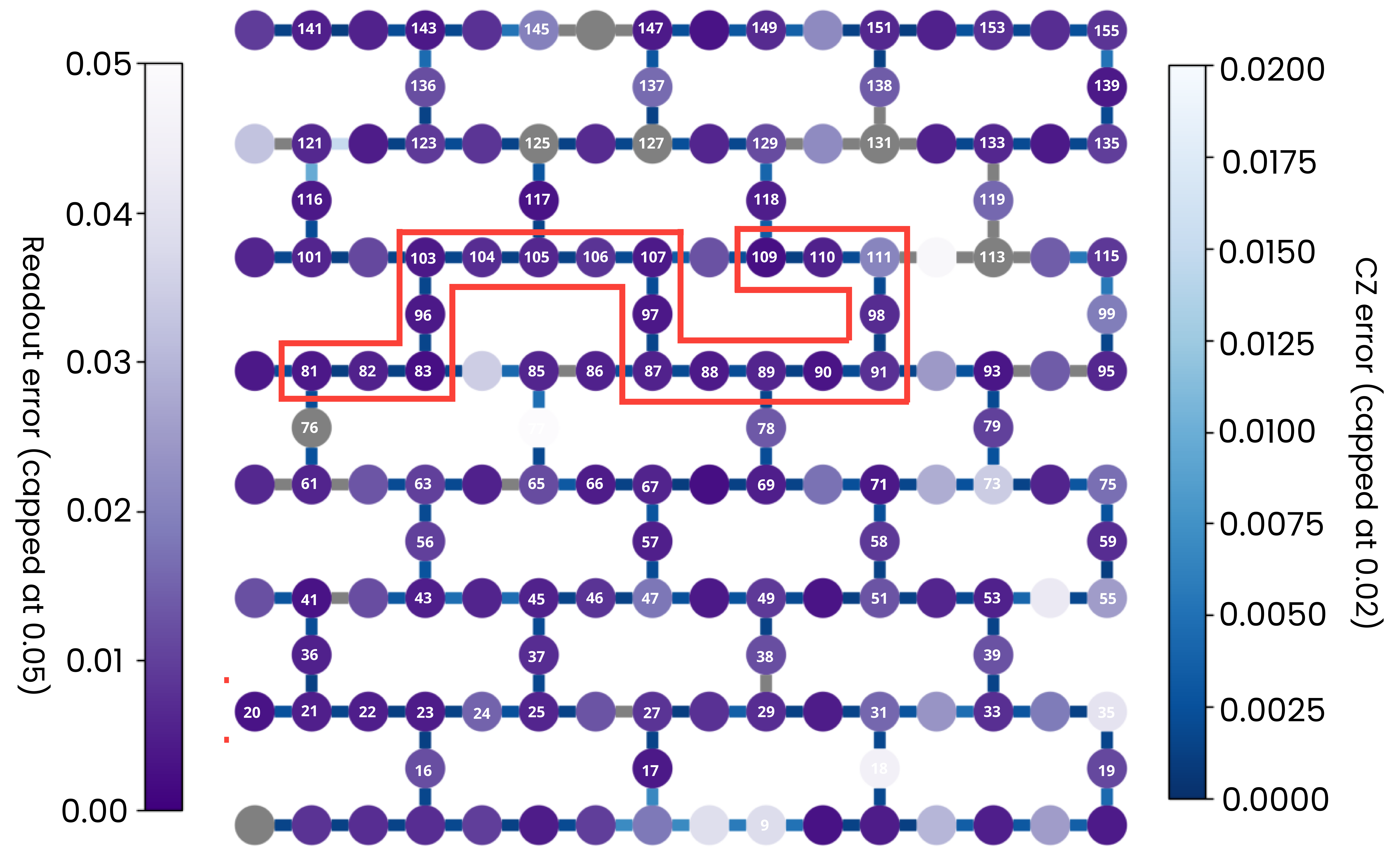}
    \caption{Same as Fig.~\ref{Mapping_kobe}, but for the \textit{ibm\_kingston} quantum processor.}
    \label{Mapping_kingston}
\end{figure}
\newpage

\end{document}